%% file: ms1.tex
\documentclass{article}
\usepackage{amssymb,amsthm,amsmath,amsfonts}
\usepackage{epsfig}

\theoremstyle{plain}
\begingroup
 
\newtheorem{theorem}{Theorem}

\endgroup

\theoremstyle{definition}
\newtheorem{remark}{Remark} 

\newcommand{\nwc}{\newcommand}
\nwc{\bit}{\begin{itemize}}
\nwc{\eit}{\end{itemize}}
\nwc{\Levy}{L\'evy}
\nwc{\LK}{L\'evy-Khintchine}
\nwc{\LI}{L\'evy-It\^{o}}
\nwc{\CH}{Cole-Hopf}
\nwc{\Holder}{H\"{o}lder}
\nwc{\Volpert}{Vol'pert}
\nwc{\cadlag}{c\`{a}dl\`{a}g}
\nwc{\BVloc}{BV_{loc}}
\nwc{\be}{\begin{equation}}
\nwc{\ee}{\end{equation}}
\nwc{\ba}{\begin{eqnarray}}
\nwc{\ea}{\end{eqnarray}}
\nwc{\la}{\label}
\nwc{\nn}{\nonumber}
\nwc{\Z}{\mathbb{Z}}
\nwc{\C}{\mathbb{C}}
\nwc{\E}{\mathbb{E}}
\nwc{\R}{\mathbb{R}}
\nwc{\prob}{\mathbb{P}}
\nwc{\PP}{\mathcal{P}}
\nwc{\PPE}{\mathcal{P}(E)}
\nwc{\M}{\mathcal{M}}
\nwc{\law}{\stackrel{\mathcal{L}}{\rightarrow}}
\nwc{\eqd}{\stackrel{\mathcal{L}}{=}}
\nwc{\vp}{\varphi}
\nwc{\Vp}{\Phi}
\nwc{\psilevy}{\Psi}
\nwc{\ve}{\varepsilon}
\nwc{\veps}{\varepsilon}
\nwc{\eps}{\ve}
\nwc{\cl}{c\'{a}dl\'{a}g}
\nwc{\qref}[1]{(\ref{#1})}
\nwc{\D}{\partial}
\nwc{\dnto}{\downarrow}
\nwc{\argmin}{\mathrm{arg}^+\mathrm{min}}
\nwc{\argmax}{\mathrm{arg}^+\mathrm{max}}
\nwc{\gena}{\mathcal{A}}
\nwc{\genb}{\mathcal{B}}
\nwc{\genf}{\mathcal{F}}
\nwc{\testfn}{\varphi}
\nwc{\jumpex}{n_*}
\nwc{\jump}{n}
\nwc{\Jbar}{j}
\nwc{\Kbar}{k}
\nwc{\Ai}{\mathrm{Ai}}
\nwc{\Painleve}{Painlev\'{e}}
\nwc{\pot}{U}
\nwc{\drift}{b}
\nwc{\den}{F}
\nwc{\velu}{V_y}
\nwc{\velv}{V_z}
\nwc{\genadag}{\gena^\dagger}
\nwc{\genbdag}{\genb^\dagger}
\nwc{\up}{u_+}
\nwc{\um}{u_-}
\nwc{\umh}{u_{-h}}
\nwc{\psip}{\psi_+}
\nwc{\psimh}{\psi_{-h}}
\nwc{\onefn}{$1$-point function}
\nwc{\twofn}{$2$-point function}
\nwc{\driftab}{\nu}
\nwc{\collab}{\sigma}
\nwc{\QQ}{\mathcal{Q}}
\nwc{\My}{\mathcal{M}}
\nwc{\Mf}{\mathcal{N}}
\nwc{\id}{\mathrm{Id}}

\begin{document}

\title{Kinetic theory and Lax equations for shock clustering and Burgers turbulence\\} 
\author{Govind Menon\textsuperscript{1} and Ravi Srinivasan\textsuperscript{2}}

\date{\today}

\maketitle

\begin{abstract}
We study shock statistics in the scalar conservation law $\partial_t u + \partial_x f(u)=0$, $x \in \R$, $t>0$, with a convex flux $f$ and spatially random initial data. We show that the Markov property (in $x$) is preserved for a large class of random initial data (Markov processes with downward jumps and derivatives of \Levy\/ processes with downward jumps).   The kinetics of shock clustering is then described completely by an evolution equation for the generator of the Markov process $u(x,t)$, $x\in \R$. We present four distinct derivations for this evolution equation, and show that it takes the form of a Lax pair.   The Lax equation admits a spectral parameter as in~\cite{Manakov}, and has remarkable exact solutions for Burgers equation ($f(u)=u^2/2$). This suggests the kinetic equations of shock clustering are completely integrable.
\end{abstract}



\medskip
\noindent
\footnotetext[1]
{Division of Applied Mathematics, Box F, Brown University, Providence, RI 02912.
Email: menon@dam.brown.edu}
\footnotetext[2]{Department of Mathematics,  The University of Texas at Austin,
Austin, TX 78712, 
Email: rav@math.utexas.edu}

\input{intro}
\input{kinetic}

\input{closure}
\input{bv}
\input{hopf}
\input{sss}

\section{Acknowledgements}
This material is based upon work supported by the National Science Foundation
under grants DMS 06-05006 and DMS 07-48482. We have benefited greatly from discussions with Mark Adler, Percy Deift, Irene Gamba, John Hartigan, David Kinderlehrer,  Dave Levermore, Mokshay Madiman, Jonathan Mattingly, Scott McKinley, Irina Nenciu, Christian Pfrang and David Pollard. We thank them all. Most of all, we thank Bob Pego for the pleasure of continued collaboration.  Part of this work was completed while the first author was visiting the Indian Statistical Institute, Kolkata.

\bibliographystyle{siam}
\bibliography{ms1}
\end{document}

%% file: intro.tex
\section{Introduction}
\label{sec:intro}
We consider the scalar conservation law
\be
\label{eq:sc1}
\partial_t u + \partial_x f(u) =0, \quad x \in \R, t>0, \quad u(x,0)=u_0(x),
\ee
with a strictly convex, $C^1$ flux $f$ and initial data $u_0$ that is a stochastic process in $x$.  The basic model of this type was introduced by Burgers in his study of turbulence. He considered $f(u)=u^2/2$ and white noise initial data~\cite{Burgers,vonNeumann,Woc}.  While this model fails to describe turbulence in incompressible fluids, it still serves as a widely useful benchmark for theoretical methods and computations in turbulence. It also has fascinating links with combinatorics, mathematical physics and statistics, some of which we describe below.
  
Our main contribution in this article is to develop a consistent kinetic theory that describes completely the clustering of shocks and the statistics of the random process $u(x,t)$ for convex $f$, and initial data $u_0$ that are Markov processes in $x$ with only downward jumps. Our approach includes exact solutions to Burgers turbulence as an important special case. In addition to being comprehensive, our approach clarifies the essential features of the problem, and reveals the role of a Lax pair that describes the shock statistics. In order to describe these results and place them in the context of past work, let us first describe the structure of solutions to \qref{eq:sc1} for deterministic $u_0$, and then some important exact solutions to Burgers turbulence that motivated our work.

\subsection{The Hopf-Lax formula} Let us first recall the notion of the {\em entropy solution\/} to \qref{eq:sc1}~\cite{Dafermos,Lax-book}. Characteristics  for \qref{eq:sc1} are lines in space-time along which $u$ is constant. If a unique characteristic connects $(a,0)$ to $(x,t)$ then $u(x,t)=u_0(a)$. However, there may be many characteristics that pass through $(x,t)$. Well-posedness is resolved by adding a small dissipative term $\eps u_{xx}$ to the right-hand side of \qref{eq:sc1} and passing to the limit $\eps \dnto 0$. This  was first carried out by Hopf in his pioneering work on Burgers equation~\cite{Hopf}, and later generalized to convex $f$ by Lax~\cite{Lax}. Let $f^*(s) = \sup_{u \in \R } \left(us -f(u)\right)$ denote the Legendre transform of $f$ and call $\pot_0(s) =\int_0^s u_0(r) dr$ the {\em initial potential}. We define the {\em{Hopf-Lax functional}}
\begin{equation}
\label{eq:hl1} 
I ( s ; x, t ) = \pot_0 ( s ) + t f^* \left(
  \frac{x - s}{t} \right).
\end{equation}
The `correct' characteristic through $(x,t)$ is given by the variational principle
\begin{equation}
  \label{eq:hl2}
a ( x, t ) = \sup \left\{ s \in
  \mathbb{R} : I ( s ; x, t ) = \inf_{r \in \mathbb{R}} I ( r ; x, t )
  \right\},
\end{equation}
We will always assume that $\pot_0$ has no upward jumps and $I$ satisfies the growth condition 
\be
\label{eq:hlgrowth}
\lim_{|s| \to \infty} I(s; x,t) = \infty.
\ee
This ensures the infimum of $I$ is a minimum, and that $a(x,t)$ is finite.  We write
\be
\label{eq:hl3}
a (x,t) = \argmin_{s \in \mathbb{R}} I ( s ; x, t ) . 
\ee
This is the {\em Hopf-Lax formula\/} for the {\em inverse Lagrangian function\/} $a(x,t)$. The $^+$ in \qref{eq:hl3} denotes that we choose $a(x,t)$ to be the largest location if $I$ is minimized at more than one point. Of particular importance is Burgers equation with $f(u)=u^2/2$. In this case, \qref{eq:hl3} is called the Cole-Hopf formula, and takes the form
\be
\label{eq:hl-burgers}
a(x,t)=\argmin_{s \in \R}\left\{\pot_0(s) + \frac{(x-s)^2}{2t}\right\}, \quad u(x,t)=\frac{x-a(x,t)}{t}.
\ee

For fixed $t$, $a$ defined by \qref{eq:hl3} is non-decreasing in $x$. If $x$ is a point of continuity of $a(\cdot,t)$, the velocity field is given implicitly by 
\be
\label{eq:hl4}
f'\left(u(x,t)\right) = \frac{x-a(x,t)}{t}.
\ee
$u(x,t)$ is well-defined because $f'$ is continuous and strictly increasing. In particular, if $a(x,t)$ is constant for $x$ in an interval, we obtain a {\em rarefaction wave}. $a$ may also have upward jumps. These arise if the minimum of $I$ is attained at more than one point. Jumps in $a$ give rise to {\em shocks} in $u$.  The left and right limits $u_\pm = u(x_\pm,t)$ exist and the velocity of the shock is given by the {\em Rankine-Hugoniot condition\/}
\be
\label{eq:hl5}
u(x,t) = \frac{f(\um)-f(\up)}{\um-\up} =:[f]_{\um,\up}.
\ee
We stress that the entropy solution $u(x,t)$ has only {\em downward jumps} in $x$. This will play a key role in our analysis.

\subsection{Exact solutions for Burgers turbulence}
Let us now suppose $u_0$ is a stochastic process in $x$. The solution develops shocks that are separated by rarefaction waves (see for example, the computations in~\cite{She}). The shocks move with speeds given by the Rankine-Hugoniot condition, and cluster when they meet. The Lax equation we derive describes this process. In order to explain this, we first focus on two exact solutions for Burgers turbulence. 

\subsubsection*{\Levy\/ process initial data on a half-line}
\label{subsec:brownian}
We assume $f(u)=u^2/2$ and 
\be
\label{eq:bminit}
u_0(x) = \left\{ \begin{array}{ll} 0, & x \leq 0, \\ X_x, & x >0, \end{array} \right.
\ee
where $X$ is a \Levy\/ processes with only downward jumps (a {\em spectrally negative\/} \Levy\/ process). A particularly interesting case is when $X$ is a standard Brownian motion (that is, a Brownian motion with $\E (X_x^2)=x$). This problem was solved formally by Carraro and Duchon~\cite{Duchon0,Duchon1} and rigorously by Bertoin~\cite{B_burgers}. The key to their solution is a {\em closure property\/} of \qref{eq:hl-burgers}: if $u_0(x)$, $x>0$ is a \Levy\/ processes with only downward jumps, then so is $u(x,t)-u(0,t)$, $x>0$. These processes are characterized by their Laplace exponent $\psi(q,t)$ defined by
\be
\label{eq:laplace}
\E\left(e^{q\left(u(x,t)-u(0,t)\right)}\right) = e^{x \psi(q,t)}, \quad x,q,t >0.
\ee
Thus, the problem is reduced to determining the evolution of $\psi$. Remarkably, $\psi(q,t)$ satisfies Burgers equation in the new variables $q,t$!
\be
\label{eq:burgers}
\partial_t \psi+ \psi \partial_q \psi =0, \quad \psi(q,0)=\psi_0(q).
\ee
If $X$ is a Brownian motion, $\psi_0(q)=q^2$, and we obtain the self-similar solution
\be
\label{eq:burgers-sol}
\psi(q,t)= \frac{1}{t^2}\psi_*(qt), \quad \psi_*(q) = q +\frac{1}{2}-\sqrt{q + \frac{1}{4}}.
\ee
Various explicit formulas are summarized in \cite{MP4}. 

The interpretation of \qref{eq:burgers} in terms of shock clustering is as follows. $\psi$ satisfies the celebrated \LK\/ formula
\be
\label{eq:LK}
\psi(q,t) = \int_0^\infty \left( e^{-qs} -1 + qs \right) \Lambda_t(ds), \quad q \geq 0, t>0,
\ee
where $\Lambda_t(ds)$ is the jump measure at time $t$. The evolution of $\psi$ by \qref{eq:burgers} also induces evolution of the jump measure $\Lambda_t$. But the jumps in the process $u(\cdot,t)$ are the shocks, which evolve in a simple manner: shocks move at speed given by the Rankine-Hugoniot condition \qref{eq:hl5} and stick, conserving momentum, when they meet. Equation \qref{eq:burgers} captures this process. It is equivalent to the fact that $\Lambda_t$
satisfies a kinetic equation, {\em Smoluchowski's coagulation equation with additive kernel}, which describes the evolution and coalescence of shocks. A derivation of the kinetic equation from this perspective may be found in~\cite{MP4}. An excellent survey of several links between stochastic coalescence and Burgers turbulence is~\cite{B_icm}.

\subsubsection*{White noise initial data}
Here we must characterize the law of $u(x,t)$ when the initial velocity is white noise. Precisely, let us suppose that the initial potential $\pot_0(x) = \sigma B_x$ where $B$ is a standard two-sided Brownian motion pinned at the origin and $\sigma$ a fixed scale parameter.  This problem  also arises in statistics, and it was in this context that Groeneboom first characterized the law of the process $u(x,t)$, $x \in \R$~\cite{Groeneboom}.  He showed that for every $t>0$, $u(x,t)$, $x \in \R$ is a Markov process with only downward jumps (a {\em spectrally negative Markov process\/}) that is stationary in $x$. He then computed the generator of this Markov process explicitly in terms of Airy functions.  Here is a brief summary of his solution.

By \qref{eq:hl-burgers} and the scaling invariance of Brownian motion, we see that 
\be
\label{eq:bscaling}
a_\sigma(x,t) \eqd (\sigma t)^{\frac{2}{3}} a_1\left(x(\sigma t)^{-\frac{2}{3}},1\right), \quad u_\sigma(x,t) \eqd \sigma^{\frac{2}{3}}t^{-\frac{1}{3}}u_1\left(x(\sigma t)^{-\frac{2}{3}},1\right),
\ee
where $\eqd$ denotes equality in law and the subscript $\sigma$ refers to the variance of $\pot_0$.  It is simplest to state the formulas under the assumption that $\sigma^2=1/2$. Let $u_x$ denote the process $u_{2^{-1/2}}(x,1)$. The  generator of  $u_x$ is an operator defined by its action on a test function $\testfn$ in its domain as follows:
\be
\label{eq:gendefn}
\gena \testfn (y) = \lim_{x \dnto 0} \frac{ \E_y \left(\testfn \left(u_x\right)\right)-\testfn(y)}{x},
\ee
where $\E_y$ denotes the law of the process with $u_0=y$. Groeneboom showed that $\gena$ is an integro-differential operator of the form
\be
\label{eq:groen1}
\gena \testfn (y) = \testfn'(y) + \int_{-\infty}^y \left( \testfn(z) -\testfn(y)\right) \jumpex(y,z) \, dz. 
\ee
The jump density $\jump_*$ of the integral operator is given explicitly as follows.
\be
\label{eq:groen2}
\jumpex(y,z) = \frac{J(z)}{J(y)}K(y-z), \quad y >z,
\ee
where $J$ and $K$ are positive functions defined on the line and positive half-line respectively, whose Laplace transforms
\be
\label{eq:groen3}
\Jbar(q) =\int_{-\infty}^\infty e^{-qy} J(y)\,dy, \quad \Kbar(q) =\int_0^\infty e^{-qy} K(y)\, dy,
\ee
are meromorphic functions on $\C$ given by 
\be
\label{eq:groen4}
\Jbar(q) =\frac{1}{\Ai(q)}, \quad \Kbar(q)= -2 \frac{d^2}{dq^2}\log \Ai(q).
\ee
$\Ai$ denotes the first Airy function as defined in \cite[10.4]{AS}. Our normalization of $J$ and $K$ differs from~\cite{Groeneboom} in two aspects. First, the above definition of $K$ is suggestive of Dyson's formula in the theory of inverse scattering~\cite[p.273]{Lax_fa}. Second, the choice $\sigma^2 =1/2$ helps avoid several factors of $2^{1/3}$ while stating the main formulas.  

We shall return to this solution at several points in this article. The one-point and two-point distribution functions can be computed once the generator is known. For example, $p(y) dy= P(u_x \in (y,y+dy)) = J(y)J(-y)dy$. The distribution of the shock sizes and the two-point distribution is given in~\cite{Frachebourg}. Tauberian arguments yield precise asymptotics of these distributions.

\subsection{Lax equations}
Despite the elegance of the solution procedure for \Levy\/ process data, it does not apply to Burgers equation with broader classes of random initial data (e.g. white noise), or to the general scalar conservation law \qref{eq:sc1}. Our goal in this article is to develop a kinetic theory
that describes shock clustering for \qref{eq:sc1}. Our work builds on the links between stochastic coalescence and Burgers turbulence~\cite{B_icm,MP4}, exact solutions to Burgers equations with white noise initial data~\cite{Frachebourg,Groeneboom}, and recent work on kinetic equations for Burgers turbulence~\cite{CD}. We amplify these remarks briefly.

Groeneboom used the variational principle \qref{eq:hl-burgers} at $t=1/2$ to compute the generator as in \qref{eq:groen1}--\qref{eq:groen3}. Girsanov's theorem is crucial in the analysis. 
Most past work in the turbulence community is also based on a similar point of view. This begins with Burgers' analysis which was eventually completed by Frachebourg and Martin~\cite{AE,Burgers,Frachebourg,Kida}. However, the focus on white noise initial data and the use of Girsanov's theorem obscures an understanding of the dynamic process of shock clustering for general initial data (see for example, the concluding remarks in~\cite{B_wn}). 

The idea of seeking kinetic equations for shock evolution can be found in pioneering early work of Burgers~\cite{B_corra,B_corrb}, but this approach remained undeveloped for several decades. More recently, E and Vanden-Eijnden introduced a hierarchy of master equations for the statistics of velocities and velocity gradients in Burgers turbulence~\cite{EvE}. However, these equations are not closed, and a sophisticated dimension reduction is needed to extract scaling exponents from these equations. Frachebourg {\em et al} showed that the hierarchy of $n$-point functions for ballistic aggregation can be closed at the level of the $2$-point function~\cite{FPM}. Most recently, Chabanol and Duchon showed formally that if a statistical solution to Burgers equation preserves the Markov property, then one can derive evolution equations for the generator of the Markov process~\cite{CD}. Moreover, they showed that Groeneboom's solution yields a self-similar solution to this evolution equation. 

This is the starting point for our work. Our viewpoint is as follows: rather than seek a specific exact solution to Burgers equation with specific initial data, we look for a class of natural stochastic processes whose structure in $x$ is preserved by the entropy solution to the general scalar conservation law~\qref{eq:sc1}. Here it is the class of {\em spectrally negative Markov processes (in $x$)}. This reduces the study of shock statistics to an evolution equation for the generator of the Markov process. Our main results are:
\begin{itemize}
\item {\em a closure theorem\/}: Suppose $f$ is strictly convex and $C^1$. If $u_0$ is a spectrally negative Markov process in $x$, so is the entropy solution to \qref{eq:sc1} (see Theorems~\ref{thm:Closure} and~\ref{thm:ClosureNoise}). This provides some  rigorous justification for our approach.

\item {\em new kinetic equations\/}: We derive kinetic equations that describe the evolution of the generator of the Markov process for arbitrary convex $f$. The kinetic equations are equivalent to a Lax pair.  Moreover, the Lax equations admit a {\em spectral parameter\/} as in Manakov's integration of the Euler equations for a spinning top in $\R^n$~\cite{Manakov}. This provides strong evidence that the kinetic equations are {\em completely integrable\/}.
\item {\em consistent derivations\/}: We derive the Lax equation from four different perspectives. Aside from being a stringent test on the consistency of the kinetic equations, this also provides a unified treatment of disparate methods in the literature on statistical hydrodynamics.
\end{itemize}

In order to explain the Lax equation, let us assume that $u(x,t)$ is a stationary, spectrally negative Feller process in $x$ whose sample paths have bounded variation.  As in the \LK\/ formula, such a process is characterized by its generator $\gena(t)$, which acts on test functions $\testfn \in C_c^{1}(\R)$ via
\be
\label{eq:gena1}
\gena\testfn(y) = \drift(y,t) \testfn'(y) + \int_{-\infty}^y \left( \testfn(z)-\testfn(y)\right) \jump(y,dz,t). 
\ee
For fixed $t>0$, $\drift(\cdot,t) \in C(\R)$ and $\jump(y,\cdot,t)$ is a measure on $(-\infty, y)$ that satisfies $\int (1 \wedge |y-z|^2) \jump(y,dz,t) < \infty$ for each $y \in \R$. These terms correspond to the drift and jumps of the process respectively. There is no diffusion term because sample paths are solutions to a conservation law and necessarily have bounded variation~\cite{Dafermos}. This expression for $\gena$ is a special case of the general form of the generator of a Markov process, as given in Theorem 3.5.3 in~\cite{Applebaum} (see also~\cite{Stroock}, Eq. (2.1.13)).

We now introduce an operator $\genb$ associated to $\gena$ and the flux function $f$. This operator is defined by its action on test functions as follows:
\be
\label{eq:genb}
\genb \testfn (y) = -f'(y)\drift(y,t) \testfn'(y) -\int_{-\infty}^y \frac{f(y)-f(z)}{y-z }\left( \testfn(z)-\testfn(y)\right) \jump(y,dz,t). 
\ee
One of our main results is that the evolution of $\gena$ is given by the Lax equation
\be
\label{eq:kinetic}
\partial_t \gena = [\gena, \genb] = \gena \genb -\genb \gena.
\ee
An evolution equation for the generator was first derived in~\cite{CD} for the flux $f(u)=u^2/2$. However, the equations in~\cite{CD} were not written in the simple form above, and it was not clear if such equations could even be derived for arbitrary fluxes $f$. We show that while $f(u)=u^2/2$ is certainly special, much of the structure of the problem relies only on the convexity of $f$. 

We present four different derivations of \qref{eq:kinetic}: (1) as a compatibility condition for martingales in $x$ and $t$; (2) from kinetic theory as in \cite{MP4}; (3) using BV calculus as in~\cite{EvE}; (4) using Hopf's functional calculus as in~\cite{CD}.  The shortest (and most heuristic) of these is the first, and goes as follows.

The main observation is that we have a two-parameter random process, and  formally $\genb$ may be viewed as the `generator' in $t$. To see this, fix $x \in \R$ and consider the random process $u(x,t)$, $t>0$. Then the multipliers $-f'(y)$ and $-[f]_{y,z}$ in \qref{eq:genb} simply correspond to the evolution (in $t$) of the drift and jumps. Indeed, if the path $u(x,t)$ is differentiable at $x$ with $u(x,t)=y$, $\partial_x u=b$, then $\partial_t u = -f'(y)b$ by \qref{eq:sc1}. Similarly, if we have a shock connecting left and right states $y$ and $z$, then the shock speed is given by \qref{eq:hl5}. 

This may be understood more precisely using It\^o's formula for jump processes \cite{Applebaum}. For fixed $t$, the random process $u(x,t)$ satisfies the stochastic differential equation 
\be
\label{eq:ito1}
du(x,t) = \drift(u(x_-))dx + \int_\R (z-u(x_-))N(u(x_-),dz,dx),
\ee
where the Poisson random measure $N(y,dz,dx)$ has intensity $\jump(y,dz)\,dx$. We write the conservation law as $du(x,t) dx + df(u(x,t)) dt = 0$, and apply It\^o's formula to obtain
\be
\nonumber
du(x,t) = -f'(u(x_-))\drift(u(x_-))dt - \int_\R [f]_{u(x_-),z} (z-u(x_-))N(u(x_-),dz,dt).
\ee
This shows immediately that $\genb$ should have drift $-f'(y)b(y)$ and jump measure $-[f]_{y,z} \jump(y,dz)$. This holds rigorously when $f$ is decreasing. In this case, all shocks move to the left, $u(x_-,t)=u(x,t_-)$,  and $u$ is also Markov in $t$, with generator $\genb$ given by \qref{eq:genb}.

If $\gena$ and $\genb$ are generators, then the one-point distributions satisfy Kolmogorov's forward equations
\be
\label{eq:forkolm}
\partial_x p = \genadag p, \quad\mathrm{and}\quad\partial_t p = \genbdag p.
\ee
We now seek martingales in $x$ and $t$. To this end, fix $(x_0,t_0)$, and consider the processes $\testfn\left(u(x_0+s,t_0)\right)$ and $\testfn\left(u(x_0,t_0+s)\right)$ with $s>0$. These processes are formally martingales if $\testfn$ solves Kolmogorov's backward equations 
\be
\label{eq:kolm1}
\partial_x \testfn +\gena \testfn=0, \quad \partial_t \testfn + \genb \testfn=0,
\ee
in the domain $(x,t) \in [x_0,\infty)\times[t_0, \infty)$, and $y,z \in \R$. If the compatibility condition $\testfn_{xt}=\testfn_{tx}$ holds for a sufficiently rich class of functions $\testfn$, we obtain the general Lax equation
\be
\label{eq:kinetic2}
\partial_t \gena -\partial_x \genb =[\gena,\genb].
\ee
If the process is stationary in $x$, $\partial_x \genb$ vanishes and we obtain \qref{eq:kinetic}. In the form \qref{eq:kinetic2}, the Lax equation is akin to {\em zero curvature\/} conditions in integrable systems.

\subsection{Kinetic theory}
\label{sec:kineticintro}
When we expand the commutator in \qref{eq:kinetic}, and separate the evolution of the drift $\drift$ and the jump measure $\jump$ we obtain a kinetic equation that describes shock clustering. The drift satisfies the differential equation
\be
\label{eq:evolb}
\partial_t \drift(y,t) = -f''(y)\drift^2(y,t).
\ee
Note that the drift does not depend on the jump measure. The jump density $\jump(y,z,t) \, dz =\jump(y,dz,t)$ satisfies the kinetic equation 
\ba
\label{eq:kin26}
\lefteqn{ \partial_t \jump(y,z,t) + \partial_y\left(\jump \velu (y,z,t)\right)+ \partial_z \left( \jump\velv(y,z,t) \right)} \\
\nonumber 
&&
= Q(\jump,\jump) + \jump \left( \left([f]_{y,z}-f'(y)\right)\partial_y \drift - \drift f''(y)\right).
\ea
Here the velocities $\velu$ and $\velv$ in \qref{eq:kin26} are given by
\be
\label{eq:kin9}
\velu(y,z,t) = \left( [f]_{y,z} -f'(y)\right) \drift(y,t), \quad \velv(y,z,t) = \left([f]_{y,z}-f'(z)\right) \drift(z,t),
\ee
and the collision kernel $Q$ is 
\ba
\nonumber
\lefteqn{ Q(\jump,\jump)(y,z,t) = \int_z^y \left([f]_{y,w}-[f]_{w,z}\right)\jump(y,w,t)\jump(w,z,t)\, dw}\\
\nonumber
&& - \int_{-\infty}^z \left([f]_{y,z}-[f]_{z,w}\right) \jump(y,z,t)\jump(z,w,t)\, dw \\
\label{eq:kin24}
 && -\int_{-\infty}^y \left([f]_{y,w}-[f]_{y,z}\right)\jump(y,z,t)\jump(y,w,t)\, dw. 
\ea
We have assumed for convenience that the jump measure has a density, but the above equations extend naturally to general jump measures. 

In the next section we derive these equations from the perspective of kinetic theory. We consider single shocks and rarefaction waves as building blocks, and use this to derive a natural Boltzmann-like equation. This will yield the equations above. We then show that these kinetic equations are equivalent to the Lax equation \qref{eq:kinetic}. That calculation reflects the fact that operators of the form~\qref{eq:gena1} formally constitute a Lie algebra.

\subsection{The broader context of our work}
We conclude this introduction by connecting our work with some other problems in mathematical physics and statistics.  We pay particular attention to connections with integrable systems. 

\subsubsection{The spectral curve and complete integrability}
We first point out the role of a spectral parameter in analogy with Manakov's treatment of the Euler equations for geodesic flow on $so(n)$ with a left-invariant metric~\cite{Manakov}. Let $\My$ and $\Mf$ denote multiplication operators acting on the domain of $\gena$, defined by
\be
\label{eq:mult_op}
\My \testfn (y) = y \testfn (y), \quad \Mf \testfn (y) = f(y) \testfn(y).
\ee
It is clear that $\My$ and $\Mf$ are diagonal operators. We now use the definitions \qref{eq:gena1}, \qref{eq:genb} and \qref{eq:mult_op} to find
\be
\label{eq:multop2}
[\gena, \Mf] -[\My,\genb]=0.
\ee
This observation allows us to introduce a spectral parameter $\mu \in \C$ in the Lax equation. We use \qref{eq:kinetic} and \qref{eq:multop2} to obtain
\be
\label{eq:multop3}
\partial_t \left(\gena -\mu \My\right) = [\gena -\mu \My, \genb +\mu \Mf], \quad \mu \in \C.
\ee
If $\gena$, $\genb$ were $n\times n$ matrices, it would follow that the spectral curve (Riemann surface)
\be
\label{eq:multop4}
\Gamma= \{ (\lambda,\mu) \in \C^2 \left| \det(\gena - \lambda \id -\mu \My) =0 \right.   \},
\ee
is fixed by the evolution. In Manakov's work, this is the crucial observation that yields the existence of additional integrals for Euler's equations in $so(n)$, $n \geq 4$. These integrals are simply the coefficients of the characteristic polynomial above. 

More broadly, the observation that \qref{eq:kinetic} admits a spectral parameter reveals a close relation with a large class of completely integrable systems (including KdV, the Toda lattice, geodesic flows on $so(n)$ and ellipsoids, and the integrable PDEs of random matrix theory). The complete integrability of all these flows may be obtained in a unified way via a general splitting theorem for Lie algebras~\cite{AdvM1}. This connection also sets the stage for the application of powerful methods from algebraic geometry to integrate \qref{eq:kinetic2} explicitly for {\em every\/} convex $f$~\cite{AdvM2}. We will address this in later work.

\subsubsection{Burgers turbulence and random matrices}
The solution to Burgers turbulence with spectrally negative \Levy\/ process data (see \S~\ref{subsec:brownian}) is obtained from \qref{eq:kinetic} as follows. Suppose $u(x,t)-u(0,t)$, $x>0$ is a spectrally negative \Levy\/ process with bounded variation and  mean zero. If we denote the jump measure $\Lambda_t$, then $\drift(t)=\int_0^\infty s \Lambda_t(ds)$, and the generators $\gena$ and $\genb$ take the form
\ba
\gena(t) \testfn(y) &= &\int_0^\infty \left(\testfn(y-s)-\testfn(y) + s \testfn'(y)\right)\Lambda_t(ds), \\
\nonumber 
\genb(t) \testfn(y)  & = & -y\gena(t) \testfn(y) +\frac{1}{2} \int_0^\infty \left(\testfn(y-s)-\testfn(y)\right) s\Lambda_t(ds).
\ea
In particular, when $\testfn(y)=e^{qy}$, $\mathrm{Re}(q)>0$, we have
\be
\label{eq:burgers-levy}
\gena(t) e^{qy} =\psi(q,t) e^{qy}, \quad \genb(t)e^{qy}=-\left(y\psi(q,t) +\frac{1}{2}\partial_q \psi\right) e^{qy}. 
\ee
We substitute \qref{eq:burgers-levy} in \qref{eq:kinetic} to obtain \qref{eq:burgers}. 

We next note that the solution \qref{eq:burgers-sol} can be mapped to Wigner's semicircle law in the theory of random matrices. Dyson observed that the eigenvalues of a standard matrix valued Brownian motion $M_t$ in the group of $n\times n$ symmetric, hermitian or symplectic matrices satisfy the stochastic differential equation~\cite{Dyson}
\begin{equation}
\label{dyson} d\lambda_k = \sum_{j \neq  k}\frac{dt}{\lambda_k-\lambda_j} + \sqrt{\frac{2}{\beta}}dB_k, \quad 1\leq k \leq n,
\end{equation}
where $B_k$, $1 \leq k \leq n$ are independent Brownian motions, and $\beta=1$, $2$ or $4$ for the ensembles above.  That is, the eigenvalues behave like repulsive unit charges on the line perturbed by independent white noise. The law of large numbers for this ensemble is  as follows. As $n \to \infty$ the spectral measure of $n^{-1/2}M_t$ converges to Wigner's semicircle law: 
\be
\label{eq:wigner}
\mu_t (dx)= \frac{1}{2\pi t}\sqrt{4t-x^2}\, dx, \quad |x| < 2\sqrt{t}.
\ee
Moreover, the Cauchy transform of $\mu_t$,
\be
\label{eq:wigner2}
g(z,t) = \int \frac{1}{z-x}\mu_t(dx), \quad z \in \C\backslash[-2\sqrt{t},2\sqrt{t}],
\ee
solves Burgers equation with a simple pole as initial data~\cite[\S 4.3.2]{Guionnet}. That is,  
\be
\label{eq:wigner3}
\partial_t g + g \partial_z g =0, \quad g(z,0) =\frac{1}{z}.
\ee
More precisely, $g$ solves \qref{eq:wigner3} in the slit plane $\C\backslash[-2\sqrt{t},2\sqrt{t}]$, and has the form
\be
\label{eq:voic}
g(z,t) = \frac{1}{\sqrt{t}}g_*\left(\frac{z}{\sqrt{t}}\right), \quad g_*(z) = \frac{1}{2}\left(z - \sqrt{z^2-4}\right), \quad |z| \geq 2.
\ee
It now transpires that the self-similar solution \qref{eq:burgers-sol} can be transformed to \qref{eq:voic} by a simple change of variables. 
\be
\label{eq:wigner4}
\frac{\psi_*(q)}{q} = g_*(z), \quad z= 2+\frac{1}{q}, \quad\mathrm{or}\quad \frac{g(t^{\frac{1}{2}}z,t)}{t^{\frac{1}{2}}} = \frac{\psi(t^{-1}q,t)}{t^{-1}q}.
\ee
 
In his beautiful thesis~\cite{Kerov}, Kerov found a deeper interpretation of \qref{eq:wigner3} based on the representation theory of the symmetric group. He introduced a Markov process for the growth of Young diagrams ({\em Plancherel growth\/}), and derived \qref{eq:wigner3} in a mean-field limit. More general initial conditions may also be included in \qref{eq:wigner3} and he showed that the evolution of $g$ by Burgers equation is equivalent to a kinetic equation for $\mu_t$~\cite[Ch. 4.5]{Kerov}. Thus, the transformation \qref{eq:wigner4} links $\mu_t$ to $\Lambda_t$, and Plancherel growth with Smoluchowski's coagulation equation (see \qref{eq:LK}). We do not have a deeper (i.e. stochastic process) explanation of this relation yet. 

Airy functions and the \Painleve\/ transcendents arise in the scaling limit of fluctuations from Wigner's law at the edge of the spectrum. The fluctuations are given by the celebrated Tracy-Widom distributions involving a solution to \Painleve-II~\cite{TW1}. We now point out that the function $l=\Jbar'/\Jbar$ ($\Jbar$ as in \qref{eq:groen3}) solves the Riccati equation $dl/dq=-q+l^2$ and is therefore an Airy solution to \Painleve-II~\cite[Ch. 7]{Ablowitz}. This is used to verify that Groeneboom's solution satisfies the kinetic equation \qref{eq:kinetic} in \S\ref{sec:sss}.

\subsubsection{Shell models of turbulence and their continuum limits}
It is also of interest to consider \qref{eq:sc1} on the half-line $x>0$ with random forcing at $x=0$. This problem arises as a continuum limit of shell models of turbulence~\cite{Mattingly-McK,Mattingly}. Shell models are lattice equations of the form $\dot{c_k}=J_{k-1} -J_{k}$, $k =1,2,\ldots$. Here $c_k \geq 0$ models the energy in the $k$-th Fourier mode (shell), and $J_k$ the flux from shell $k$ to $k+1$. $J_k$ is expressed in terms of $c_k$ and its nearest neighbors by a suitable constitutive relation. The main question is to understand how randomness spreads through the system under the assumption that $J_0$ is a prescribed random forcing (a stationary Feller process, for example). In particular, it is of interest to understand whether the system has a unique invariant measure. In the continuum limit, these questions may be treated by our approach. We note that the derivation of \qref{eq:kinetic2} does not depend on assumptions of stationarity and is independent of boundary conditions. Therefore, \qref{eq:kinetic2} holds in the domain $x,t>0$ and must be augmented with a boundary condition on the line $x=0$. In particular, as $t\to \infty$, invariant measures are solutions to 
\be
\label{eq:kinetic3}
-\partial_x \genb =[\gena, \genb],
\ee
with a boundary condition at $x=0$ matching $\genb$ and the generator of the forcing. 

In contrast to a time-correlated boundary forcing, it is also possible to consider a white-in-time forcing in the bulk (a Feller process in $x$ and $\delta$-correlated in $t$) which is independent of the initial data. For Burgers' equation, this is a particular case of {\em forced Burgers turbulence\/}. The generator $\genf$ of the forcing simply appears as an additional term in the Lax equation
\be
\label{eq:kinetic4}
\partial_t \gena -\partial_x \genb =[\gena, \genb] + \genf.
\ee
An interesting exact solution to \qref{eq:kinetic4} for  Burgers equation forced with a two-sided Brownian motion in space has recently been obtained in \cite{CD2}.

\subsubsection{Applications to statistics}

In statistics, \qref{eq:hl-burgers} with $\pot_0$ a two-sided Brownian motion first arose in the following estimation problem~\cite{Chernoff}. Suppose $X_1, \ldots, X_n$ are independent, identically distributed (iid) samples from a distribution with a smooth unimodal density $\rho$ with mode $m$ and finite variance. Let us consider a naive `binning' strategy to estimate the mode $m$. We fix a bin width $w$ and count the number of samples $N_n(s) = \#\{X_k \in (s-w,s+w)\}$ in the bin centered at $s$. We estimate the mode by $m_n = \argmax_s N_n(s)$. Chernoff observed that the fluctuations $m_n-m$ are $O(n^{-1/3})$. Precisely, for suitable $c(\rho,w)>0$, the rescaled random variables $cn^{1/3}(m_n-m)$ converge in law to {\em Chernoff's distribution}
\be
\label{eq:chernoff}
Z=\argmax_s\left\{ \pot_0(s) - s^2\right\}.
\ee
The quadratic term $-s^2$ arises from the Taylor expansion of  $\rho$ around $m$ -- the
first order term vanishes since $m$ is the maximum of $\rho$. By the symmetry of Brownian motion, it is clear that $Z$ has the same law as $a(0,1/2)$. 

This is not an isolated result: `cube-root' fluctuations appear naturally in a wide class of estimation problems~\cite{Tukey}. Kim and Pollard proved functional limit theorems for several such estimators, with the law of the limit characterized by 
$\argmax_{s \in \R^d} \left\{ \pot_0(s) - |s|^2\right\}$ where $\pot_0(s)$ is a continuous Gaussian process in $\R^d$ pinned at the origin~\cite{Kim-Pollard}. These correspond to solutions to Burgers equation in $\R^d$ with random initial data $U_0$, but this connection has not been explored in the literature.

\subsection{Outlook}
Let us conclude by pointing out some significant shortcomings in our work. Much of this article relies on formal calculations. But these calculations are often interesting, and it seems more fruitful to present them in a transparent and suggestive manner, rather than as rigorous statements burdened by technicalities. The most significant gap is that we do not {\em prove} \qref{eq:kinetic}. This is because our closure theorem is not strong enough. We only establish that the entropy condition preserves the Markov property. In order to rigorously establish \qref{eq:kinetic} it is necessary to prove that the entropy condition preserves Feller processes. This is an assertion of regularity, whereas we only establish measurability. This is closely tied to establishing a 
satisfactory well-posedness theory for \qref{eq:kinetic}, and will be addressed in forthcoming work. A suitable well-posedness theorem would also yield a probabilistic proof of the existence of a two-parameter family of self-similar solutions  to \qref{eq:kinetic}. These solutions are generated by considering the flux functions $f(u)= |u|^p/p$, $1<p<\infty$ and an initial potential that is an $\alpha$-stable spectrally negative \Levy\/ process. It seems challenging to prove this analytically starting with \qref{eq:kinetic}.

Also, while our work yields a deeper understanding of Groeneboom's solution, it does not, as yet, constitute an independent proof of his results. We have only been able to verify that his solution satisfies \qref{eq:kinetic}. We have been unable to derive it using \qref{eq:kinetic} alone. We hope to address this in future work using techniques from integrable systems.

The rest of this article is organized as follows. We first derive  \qref{eq:evolb} and \qref{eq:kin26} from the standpoint of kinetic theory in the next section.  This is followed by the rigorous closure theorems.  We then  
derive the Lax equations by BV calculus and by Hopf's method. Finally, we consider Groeneboom's self-similar solution for Burgers equation in Section~\ref{sec:sss}.

%% file: kinetic.tex
\section{Kinetic equations}
\label{sec:kinetic}
\subsection{Introduction}
In this section, we assume the velocity $u$ is a stationary Feller process in $x$, and derive the kinetic equations of Section~\ref{sec:kineticintro}. We use the evolution of a single shock and rarefaction wave to derive a Boltzmann-like equation for the evolution of the density of shocks. We conclude this section by showing that the Lax equation \qref{eq:kinetic} is equivalent to the kinetic equations \qref{eq:kin26}--\qref{eq:kin24}. The main observation is that the space of operators of the form \qref{eq:gena1} is formally a Lie algebra.

\subsection{Conservation of total number density}
Let $p(y,t)$ denote the stationary $1$-point density, i.e. $p(y,t)dy = P(u(x,t)\in(y,y+dy))$, and $\den(y,z,t)$ denote the total number density, i.e. the expected number of jumps per unit length from states $y$ to $z$. Then 
\be
\label{eq:kin1}
\den(y,z,t) = p(y,t) \jump(y,z,t).
\ee
The total number density changes because of a flux of shocks and shock collisions, and we have the general conservation law of Boltzmann-type
\be
\label{eq:kin2}
\partial_t \den + \partial_y\left( \den \velu\right) + \partial_z\left( \den \velv\right) = C(\den,\den).
\ee
Here $\velu$ and $\velv$ denote `velocities' in the $(y,z)$ `phase space', and $C$ is a binary collision kernel. This is the general structure of the equation. We now derive the evolution equation for $\drift$, the velocities $\velu$, $\velv$, and the collision kernel $C$, based on elementary solutions to the scalar conservation law \qref{eq:sc1}.

\subsection{Decay of the drift} First consider how affine data evolves under \qref{eq:sc1}. Let $u(x,t)$ solve \qref{eq:sc1} with $u_0(x) =\alpha_0 + \beta_0 x$. For $x,t \approx 0$, we have to leading order $u(x,t) \approx \alpha(t) +\beta(t) x$, so that
\be
\label{eq:kin3} \partial_t u \approx \dot{\alpha} + \dot{\beta}x, \quad \partial_x f(u) \approx f'(\alpha) \beta + f''(\alpha)\beta^2 x.
\ee
We now balance terms  in \qref{eq:sc1} to obtain
\be
\label{eq:kin4} 
\dot\alpha = -f'(\alpha)\beta, \quad \dot\beta = -f''(\alpha) \beta^2. 
\ee
The second equation expresses the decay of rarefaction waves when $\beta_0>0$. The connection between this elementary solution and the generator of the process is the following. The drift coefficient $\drift(y,t)$ corresponds locally to an affine profile as above with $\alpha=y$ and $\beta=\drift$. Thus, the second equation above is simply \qref{eq:evolb}. 

\subsection{Decay of shocks and $\velu$,$\velv$}
The `velocities' $\velu$ and $\velv$ in $(y,z)$ space arise  because of the decay of shocks. In order to derive these velocities, we fix $u_->u_+$ and consider piecewise affine initial data 
\be 
\label{eq:kin4a}
u_0(x) = \left\{ \begin{array}{l} u_- + \drift_- x, \quad x <0, \\ u_+ + \drift_+x, \quad  x >  0. \end{array} \right.
\ee
Let $s(t)$ denote the path of the shock. Then by \qref{eq:hl5}, for small $t$ 
\be
\label{eq:kin5}
s(t) = [f]_{u_-,u_+}(t + o(t)).
\ee
$s$ is also given by the kinematic condition
\be
\label{eq:kin6}
s(t) = a_\pm(t) + f'(u_\pm)t  + o(t).
\ee
where $a_\pm$ denotes the left and right inverse Lagrangian points at time $t$. Thus,
\be
\label{eq:kin7}
a_\pm(t) = \left([f]_{u_-,u_+}-f'(u_\pm) \right)t +o(t).
\ee
Since $u_->u_+$ and $f$ is convex, we see that $[f]_{u_-,u_+}-f'(u_-)<0$. Similarly,  $[f]_{u_-,u_+}-f'(u_+)>0$. As a consequence, a shock initially connecting states $u_\pm$ decays to a shock connecting states 
\be
\label{eq:kin8}
u_\pm + \left( [f]_{u_-,u_+}-f'(u_\pm)\right)\drift_\pm t  + o(t).
\ee
This decay gives rise to a flux of $\den$. To first approximation, the flux is linear in $\den$, and of the form \qref{eq:kin2} with the drift velocities given by \qref{eq:kin9}. 

\subsection{The collision kernel $C(\den,\den)$} Binary collisions of shocks occur at a rate determined by the Markov property (in $x$) and the relative velocity of shocks given by \qref{eq:hl5}. In order to simplify notation, we suppress the $t$-dependence for the process $u$ and write $u_x$ for $u(x,t)$. We also denote a shock connecting states $u_-$ and $u_+$ by $\{u_-,u_+\}$. Shock clustering involves the following events:
\ba
\label{eq:kin10}
&& \mathrm{growth:} \quad \{y,w\} + \{w,z\} = \{y,z\}, \quad z < w < y, \\
\label{eq:kin11}
&&
\mathrm{decay}: \quad \{y,z\} + \{z,w\} = \{y,w\}, \quad  -\infty < w < z, \\
\label{eq:kin12}
&&
\mathrm{decay}: \quad \{w,y\} + \{y,z\} = \{w,z\}, \quad  y < w < \infty.
\ea
The computation of rates for these events is similar. To be concrete, let us first consider \qref{eq:kin10}. Fix $z<w< y$ and consider small $\Delta x_1 >0$, $\Delta x_2 >0$. Then formally, by the Markov property
\be
\label{eq:kin12a}
dP(u_0=y,u_{\Delta x_1}=z) \approx p(y)\jump(y,z) \, dy \, dz \, \Delta x_1,
\ee
and similarly,
\ba
\label{eq:kin13}
\lefteqn{ dP\left(u_0=y, u_{\Delta x_1}=w, u_{\Delta x_1+\Delta x_2}= z\right)} \\
&&
\nonumber
\approx p(y) \jump(y,w) \jump(w,z) \,dy\, dz\, dw\, \Delta x_1\, \Delta x_2. 
\ea
The relative velocity of these shocks is $[f]_{y,w}-[f]_{w,z}$ to leading order. We thus set  $\Delta x_2= \left([f]_{y,w}-[f]_{w,z}\right)\, \Delta t$ to compute the number of collisions in time $\Delta t$. We now sum over $w$ in the range $z<w<y$ to obtain the growth term in $C(\den,\den)$:
\be
\label{eq:kin14}
C_1 := \int_{z}^y p(y)\jump(y,w) \jump(w,z) \left( [f]_{y,w}-[f]_{w,z} \right) \, dw.
\ee
The computation for the events \qref{eq:kin11} is similar. We now find the decay terms
\be
\label{eq:kin15}
C_2:= -\int_{-\infty}^z p(y)\jump(y,z) \jump(z,w) \left( [f]_{y,z}-[f]_{z,w} \right) \, dw,
\ee
and 
\be
\label{eq:kin16}
C_3 := -\int_{y}^\infty p(w)\jump(w,y) \jump(y,z) \left( [f]_{w,y}-[f]_{y,z} \right) \, dw.
\ee

\subsection{Kinetic equations for $\jump$}
We have now defined all the terms in \qref{eq:kin2}. In order to obtain an equation in terms of $\jump$ alone we use Kolmogorov's forward equations to eliminate the $1$-point distribution $p(y,t)$. The first equation in \qref{eq:forkolm} is now $0=\genadag p$ since the process is stationary in $x$. The second equation in \qref{eq:forkolm} implies
\be
\label{eq:kin17}
p(y)\partial_t \jump(y,z) = \partial_t \den(y,z) - \jump(y,z) \genbdag p(y),
\ee
where 
\be
\label{eq:kin18}
\genbdag p(y) = \partial_y \left(f'(y)b(y)p(y)\right) - \int_\R \left( p(w) \jump(w,y)-  p(y) \jump(y,w)\right)[f]_{y,w}\, dw. 
\ee
(It is convenient to denote the domain of integration by $\R$, noting that it is actually a half-line because $n(y,z)=0$ for $y<z$.)

To isolate the main cancellations on the right hand side of \qref{eq:kin17}, we note that  the integral term in $C_3 - \jump(y,z)\genbdag p(y)$ is
\be
\label{eq:kin19}
\jump(y,z) \left( [f]_{y,z} \int_\R p(w) \jump(w,y) \,dw - p(y) \int_\R \jump(y,w) [f]_{y,w} \, dw \right).
\ee
The first integral above can be simplified further. Since  $\genadag p=0$, we also have 
\be
\label{eq:kin20}
\int_\R p(w)\jump(w,y)\,dw = p(y)\int_\R \jump(y,w)\, dw + \partial_y\left(b(y) p(y)\right).
\ee
Therefore, we may rewrite the expression in \qref{eq:kin19} as 
\be
\label{eq:kin21}
p(y)\jump(y,z)\int_{-\infty}^y\left([f]_{y,z}-[f]_{y,w}\right)\jump(y,w)\, dw
+ \jump(y,z)[f]_{y,z}\partial_y\left(b(y) p(y)\right).
\ee
We now collect all terms on the right hand side of \qref{eq:kin17} using \qref{eq:kin2}, \qref{eq:kin9}, \qref{eq:kin14}, \qref{eq:kin15}, and \qref{eq:kin21}. We then have
\ba
\lefteqn{p(y)\partial_t \jump(y,z) = p(y)Q(\jump,\jump)(y,z)}\\
&&
\label{eq:kin22}
+ \jump(y,z)\left( [f]_{y,z} \partial_y\left(b(y) p(y)\right) -\partial_y \left( f'(y) b(y) p(y)\right) \right) \\
&& 
\label{eq:kin23}
- \partial_y \left(\den(y,z) \velu(y,z)\right) -\partial_z \left( \den(y,z) \velv(y,z)\right),
\ea
where $Q(\jump,\jump)$ denotes the collision kernel in \qref{eq:kin24}.
Finally, we use  \qref{eq:kin1} and \qref{eq:kin9} to obtain the sum of  \qref{eq:kin22} and \qref{eq:kin23}:
\be
\label{eq:kin25} 
= p(y) \left( \jump\left( [f]_{y,z}-f'(y)\right)\partial_y b - bf''(y)\right) -\partial_y\left(\jump \velu(y,z)\right)-\partial_z \left( \jump\velv(y,z)\right) .
\ee
Under the assumption that the Markov process $u$ has a strictly positive stationary density, we may cancel $p(y)$ on both sides of \qref{eq:kin21}. We then have the kinetic equations \qref{eq:kin26}.

\subsection{Equivalence of the Lax equations and kinetic equations}
The equivalence of the Lax equations \qref{eq:kinetic} and the kinetic equations \qref{eq:kin26}--\qref{eq:kin24} follows from the algebraic structure of operators of the form \qref{eq:gena1} and \qref{eq:genb}. The space of 
operators of the form
\be
\label{eq:alg1}
\gena \testfn(y) = \drift  (y)\testfn(y) + \int_\R \jump (y,z) \left( \testfn(z)-\testfn(y)\right) \, dz,
\ee
with $\drift$, $\jump$ smooth, $\jump$ satisfying appropriate integrability conditions, and the bracket $[\cdot,\cdot]$ is formally a Lie algebra. That is, if $\gena_i$, $i=1,2,3$ are operators as above, then $[\gena_1,\gena_2]$ is an operator of the same form, and the following Jacobi identity holds:
\be
\label{eq:jacobi}
[[\gena_1,\gena_2],\gena_3] + [[\gena_2,\gena_3],\gena_1]+[[\gena_3,\gena_1],\gena_2]=0.
\ee
We do not assume that $b$ and $\jump$ are positive or that their support is a half-line. 
Both assertions rely on tedious, but direct calculations. We will omit the proof of \qref{eq:jacobi}, and simply summarize the calculation for $[\gena_1,\gena_2]$. We have 
\be
\label{eq:alg2}
[\gena_1,\gena_2] \testfn(y) = \left(\drift_1 \drift_2' - \drift_1'\drift_2\right) \testfn(y)
+ \int_{\R} \left( \driftab + \collab \right)\left(\testfn(z)-\testfn(y)\right)\, dz, 
\ee
with
\be
\label{eq:alg3}
\driftab(y,z) = \drift_1(y) \partial_y \jump_2 - \drift_2(y) \partial_y \jump_1 + \partial_z \left(\drift_1(z) \jump_2 -\drift_2(z) \jump_1\right), 
\ee
and
\ba
\label{eq:alg4}
\collab(y,z) &&= \int_{\R} \left( \jump_1(y,w)\jump_2(w,z) - \jump_2(y,w)\jump_1(w,z) \right)\, dw\\
\nonumber
&& + \int_{\R}  \jump_2(y,z) \jump_1(z,w)- \jump_1(y,z) \jump_2(z,w) \, dw \\
\nonumber
&& + \int_{\R} \jump_1(y,z)\jump_2(y,w) -\jump_2(y,z) \jump_1(y,w) \, dw.
\ea
We apply \qref{eq:alg2} to \qref{eq:kinetic} as follows. Let $\gena_1=\gena$ and $\gena_2=\genb$, with $\gena$ and $\genb$ as in \qref{eq:gena1} and \qref{eq:genb}. We then find immediately from \qref{eq:alg3} that the drift coefficient of $[\gena,\genb]$ is 
\be
\label{eq:alg5}
-\drift(y) \left(f'(y)\drift(y)\right)' + f'(y)\drift(y)\drift'(y) = -f''(y)\drift^2(y)
\ee
as in \qref{eq:evolb}. Similarly, we use  \qref{eq:alg3} to obtain the second, third and last term in \qref{eq:kin26}. Finally, we substitute the jump measures in \qref{eq:alg4} to obtain the collision kernel $Q(\jump,\jump)$.

%% file: closure.tex
\section{Closure theorems}
\label{sec:closure}
\subsection{Introduction}
In this section we show that the entropy solution to \qref{eq:sc1} preserves the class of spectrally negative Markov processes. The Markov property of $u$ was implicitly used by Burgers, and made explicit in~\cite{AE}. Our work is based on previous results by Bertoin~\cite{B_burgers}  and Winkel~\cite{Winkel}. We extend these results to general convex fluxes $f$ and a large class of noise initial data.

\subsection{Splitting times}
The main technical tool we need is the decomposition of Markov processes at random times so that the past and future are conditionally independent. This certainly holds at a stopping time. However, it also holds for a broader class of {\em splitting times\/}.  We use the following theorem on the preservation of the Markov property at a last-passage time.
\begin{theorem}
  [Getoor~\cite{Getoor}]\label{thm:Getoor} Consider a c\`adl\`ag strong Markov process $X_s$.
  Let $M \subset \mathbb{R}$ be a fixed set and let $L = \sup \{ s \in \mathbb{R} : X_s\in M \}$
  be the end of $M$. Then the post-$L$ process $\{ X_s \}_{s \geq L}$ is independent of
  $\{ X_s \}_{s < L}$ given $X_L$. 
\end{theorem}
Furthermore, the transition semigroups of the pre- and post-$L$ processes can be explicitly determined from that of $X$. It is an important and subtle fact that these transition semigroups are different from that of the original process. In particular, the semigroup of the post-$L$ process is the same as that of $X$ conditioned to never hit $M$ again.

\subsection{Closure for spectrally negative initial velocity} 
\label{subsec:Closure}
We work within the canonical framework. Let $u_0 = ( \Omega, \mathcal{F},
\mathcal{F}_s, u_0 ( s ), \mu_0 )$ be a spectrally negative Markov process on
the space $\Omega$ of c\`adl\`ag paths on $( - \infty, \infty )$ endowed with
the Skorohod topology. Here, $\mu_0$ is the law of the coordinate process $u_0
( \omega ; s ) = \omega ( s )$ on $\Omega$ and $\mathcal{F}_s = \sigma (
u_0 ( r ), r \leq s )$ is the natural filtration of $u_0$ extended to be
right-continuous and complete. Note that in our work $x$ plays the role of ``time'' as traditionally used in the theory of stochastic processes. $t$ acts as a parameter in the following discussion, and will be dropped from the notation when convenient. 

We rewrite the Hopf-Lax functional and inverse Lagrangian function as
\begin{equation}
  \label{eq:HopfLaxFunctional} I ( s ; x, t ) = \int_0^s \left( u_0 ( r ) - (
  f' )^{- 1} \left( \frac{x - r}{t} \right) \right) d r
\end{equation}
\begin{equation}
  \label{eq:InverseLagrangian} a ( x, t ) = \arg^+ \min_{s \in
  \mathbb{R}} I ( s ; x, t ) .
\end{equation}

To ensure that (\ref{eq:InverseLagrangian}) is well-defined and finite, we assume  (\ref{eq:hlgrowth}) holds a.s. This is certainly true if $u_0$ is a stationary random process and $f^*$ grows fast enough at infinity. In what follows we also make use of the fact that for fixed $t$, $a(x,t)$ is an increasing function in $x$. Since $u ( x, t )$ and $a ( x, t )$ are related by \qref{eq:hl4}, in order to obtain the closure property we only need to show that $a ( x, t )$ is a Markov process. We prove the following:

\begin{theorem}
  \label{thm:Closure}Let $u_0$ be a spectrally negative strong Markov process such that 
 \qref{eq:hlgrowth} holds a.s. Under the law $\mu_0$
  and for any fixed $t > 0$, the inverse Lagrangian process $a ( x, t )$ is
  a Markov process. Thus, $u ( x, t )$ is a spectrally negative Markov process.
\end{theorem}

\begin{proof}
  Fix $x \in \mathbb{R}$. Without loss of generality it suffices to take $t =
  1$. We also suppress the dependence on $t$ in the notation for clarity.
The proof consists of three steps, the first two of which are entirely
  deterministic. 

 {\emph{1. Dependence structure of $( a ( x ) )_{x \in \mathbb{R}}$}}: 
Since $a ( x )$ is increasing, for $h > 0$
  \begin{align}
    a ( x + h ) & = \arg^+ \min_{s \in \mathbb{R}} I ( s ; x + h )
    \nonumber\\
    & = \arg^+ \min_{s \geq a ( x )} \{ I ( a ( x ) ; x + h ) + ( I
    ( s ; x + h ) - I ( a ( x ) ; x + h ) ) \} \nonumber\\
    & = a ( x ) + \arg^+ \min_{s \geq 0} \left\{ \int_{a ( x )}^{a (
    x ) + s} \left( u_0 ( r ) - ( f' )^{- 1} ( x - r ) \right) d r \right\} . 
    \label{eq:DependenceStructure}
  \end{align}
  $( a ( x + h ) - a ( x ) )_{h > 0}$ therefore depends on $u_0$ only through
  $( u_0 ( s ) )_{s > a( x )}$. The same argument shows that $( a ( x - h ) -
  a ( x ) )_{h > 0}$ depends only on $( u_0 ( s ) )_{s < a ( x )}$.
    
{\em 2. Downward jumps determine $u_0 ( a ( x ) )$}:
If $u_0$ has only downward jumps then $\partial_s I ( a( x ) ; x ) = 0$. That is,
  \begin{equation}
    \label{eq:LagrangianCondition} u_0 ( a ) = ( f' )^{- 1} ( x - a( x ) ) .
  \end{equation}
  This may be seen easily by sketching a picture. Upward jumps in $u_0$ give rise to a 
  potential with a corner that is convex, and downward jumps give a potential with a corner that is concave. In the second case, the minimum of $I$ can never be achieved at the corner. Thus, it is always obtained at a point of continuity, implying \qref{eq:LagrangianCondition}.

 {3. \emph{$a( x )$ is a splitting time for $u_0$}}: It is not obvious that $a( x )$
  is a Markov time for $u_0$. It is not a
  stopping time since the event $\{ a ( x ) \leq s \}$ is not
  $\mathcal{F}_s$-measurable. Indeed, $a ( x )$ is the
  {\em last time\/} (\ref{eq:HopfLaxFunctional}) is minimized and depends on
  $(u_0 ( r ) )_{r > s}$ as well. $a ( x )$ is, however, a splitting time for a
  certain functional of $u_0$. To see this, first define
  \[ m ( s ; x ) = \min_{r \leq s} I ( r ; x ), \qquad D ( s ; x ) = ( I
     - m ) ( s ; x ) . \]
  For simplicity, denote $d(s;x) = (f')^{- 1} ( x - s )$.
  Dropping the explicit dependence on $x$ from the notation, consider the
  process $( u_0, d, D )$. Since $D ( s )$ is entirely dependent on $(
  u_0 ( r ) )_{r \leq s}$, $( u_0, d, D ) ( s )$ is $\mathcal{F}_s$-measurable.
  For any $\mathcal{F}$-stopping time $\xi$
  \[ m ( \xi + s ) = m ( \xi ) \wedge \left( \min_{\xi < r \leq \xi + s} I ( r
     ) \right) = I ( \xi ) + \left\{ - D ( \xi ) \wedge \min_{0 < r \leq s} (
     I ( \xi + r ) - I ( \xi ) ) \right\} \]
  and
  \[ D ( \xi + s ) = ( I ( \xi + s ) - I ( \xi ) ) - \left\{ - D ( \xi )
     \wedge \min_{0 < r \leq s} ( I ( \xi + r ) - I ( \xi ) ) \right\} . \]
  For $s > 0$ the increment $I ( \xi + s ) - I ( \xi )$ depends on
  $\mathcal{F}_{\xi}$ only through $u_0 ( \xi )$ by the strong Markov property
  of $u_0$. Therefore, $( u_0, d, D )$ is also a c\`adl\`ag, strong Markov process.
  
  We now show that $a ( x )$ is a splitting time for $( u_0, d, D )$. By
  definition (\ref{eq:InverseLagrangian}) and step 2, $a ( x )$ is the last time $s$ that
  $( u_0, d, D )(s)$ hits the fixed set $\{(y,y,0) : y\in\mathbb{R} \}$.
  Thus, we can use Theorem \ref{thm:Getoor} to split $u_0$ at $a ( x )$ into
  \[ ( u_0, d, D )_{s < a ( x )} \text{ and } ( u_0, d, D )_{s > a ( x )}. \]
 As a consequence, 
  \[ ( u_0 ( s ) )_{s < a ( x )} \text{ and } ( u_0 ( s ) )_{s > a ( x )} \]
  are conditionally independent given $d(a(x))$---that is, since $d$ is invertible,
  given $a ( x )$.

This shows that $a$ is a Markov process (note that the law of the increments of the process may vary with $x$). Therefore, $u ( x ) = ( f' )^{- 1} ( x - a ( x ) )$ is a
  spectrally negative Markov process. 
\end{proof}
\begin{remark}
Theorem~\ref{thm:Closure} reduces to Bertoin's closure theorem~\cite[Thm. 2]{B_burgers} if $f ( u )   = u^2 / 2$ (Burgers flux) and $u_0$ is a spectrally negative L\'evy process. To see
  this, use (\ref{eq:LagrangianCondition}) in (\ref{eq:DependenceStructure}):
  \begin{align}
    a ( x + h ) - a ( x ) & = \arg^+ \min_{s \geq 0} \left\{ \int_{a
    ( x )}^{a ( x ) + s} \left( u_0 ( r ) - u_0 ( a ( x ) ) + r - a ( x )
    \right) d r \right\} \nonumber\\
    & = \arg^+ \min_{s \geq 0} \left\{ \int_0^s \left( u_0 ( a ( x )
    + r ) - u_0 ( a ( x ) ) + r \right) d r \right\} . \nonumber
  \end{align}
  Since $u_0$ has independent increments, the integrand above is independent of $(
  u_0 ( s ) )_{s < a ( x )}$ and has the same law as $u_0 ( r ) - u_0 ( 0 ) +
  r$. So, the law of $a ( x + h ) - a ( x )$ is independent of $a ( x - h ) -
  a ( x ) )_{h \geq 0}$ and does not vary with $x$. Note that it is necessary
  to use $f(u)=u^2/2$ to obtain this result, and to show that the increments of
  $a ( x )$ are identical in law to those of the first hitting process $T ( x
  ) = \inf \{ s \geq 0 : u_0 ( s ) + s \geq x \}$.
\end{remark}

\subsection{Closure for noise initial data}
We now prove the analogue of Theorem~\ref{thm:Closure} for noise initial data. Recall that the Hopf-Lax functional
and inverse Lagrangian function satisfy
\begin{equation}
  \label{eq:HopfLaxFunctional2} I ( s ; x, t ) = \pot_0 ( s ) + f^* \left(\frac{x - s}{t} \right)
\end{equation}
\begin{equation}
  \label{eq:InverseLagrangian2} a ( x, t ) = \arg^+ \min_{s \in
  \mathbb{R}} I ( s ; x, t ) .
\end{equation}
We prescribe the law of the potential as follows. Let
\begin{equation}
  \pot_0 ( x ) = \left\{ \begin{array}{cc} - \tilde{Y}_{- x_-}, & x \leq 0  \\  Y_x, & x > 0  \end{array}, \right.
\end{equation}
where $Y, \tilde{Y}$ are independent copies of a spectrally negative
additive process starting at 0 and $x_-$ denotes the limit from the left. We denote by $M_0$ the law of $\pot_0$ on the space of c\`adl\`ag paths. If $X$ has stationary increments it is 
a L\'evy process, but it is not necessary to assume this. In addition, assume the growth condition (\ref{eq:hlgrowth}) holds almost surely.

\begin{theorem}
  \label{thm:ClosureNoise}Suppose $\pot_0 = ( \Omega, \mathcal{G}, \mathcal{G}_s,
  \pot_0 ( s ), M_0 )$ is a two-sided spectrally negative process with
  independent increments and satisfies (\ref{eq:hlgrowth})
  a.s. Then under $M_0$ and for all $t > 0$, $u ( x, t )$ is a spectrally
  negative Markov process.
\end{theorem}

\begin{proof}
  Again, we fix $x \in \mathbb{R}$, take $t = 1$, and drop the $t$-dependence in our notation.
  
  1. {\emph{Dependence structure of $( a ( x ) )_{x \in \mathbb{R}}$}}: As
  before, for $h > 0$
  \begin{align}
    a ( x + h ) & = \arg^+ \min_{s \in \mathbb{R}} I ( s ; x + h )
    \nonumber\\
    & = a ( x ) + \arg^+ \min_{s \geq 0} \left\{ I ( a ( x ) + s ; x
    + h ) - I ( a ( x ) ; x + h ) \right\} \nonumber
  \end{align}
  Let $X_s$ denote the strong Markov process $I(s ; x + h)$ and let $a = a(x)$.
  The increment $a ( x + h ) - a ( x )$ therefore depends on $\pot_0$ only through
  the increment $X_{a+s} - X_a$.
    
  2. {\emph{Spectral negativity of $X$}}: Since $\pot_0$ is spectrally negative, so is
  the process $X_s$.  Explicitly,
  \[ X_{s_-} \wedge X_s = X_s . \]
    
  3. {\emph{Markov property of $X$ at $a$}}:
  The inverse Lagrangian $a$ is the last time $X_{s-} \wedge X_s $ hits its ultimate minimum $m=X_a$.
  Using spectral negativity, we need only consider $X_s$. Let $M_s = \min_{r\leq s} X_r$.
  Then $a ( x )$ is the last time $s$ that $( X_s, M_s )$ hits the fixed set $\{(y,y) : y\in\mathbb{R} \}$.
  Theorem \ref{thm:Getoor} then implies $(X_s)_{s \geq a}$ is independent of $(X_s)_{s < a }$ given $m$.

  We now make use of the particular form of the semigroup of the post-$a$ process, for which we refer
  to~\cite[Thm. 2.12]{Getoor} (see also Theorem 5.1 in~\cite{MSW} and the subsequent remark).
  Precisely, given $a$ and $m$, the law of $(X_s)_{s \geq a}$ is identical to that of an independent
  copy of $X_s$ started at $a$ with $X_a = m$ and conditioned to live above $m$. By the invariance of $U_0$ under
  translations in state space (due to the property of independent increments) this
  implies that $(X_s)_{s \geq a}$ is identical in law to a copy of $X_s$ started at $a$ with value $X_a = 0$ and
  conditioned to live above $0$. To summarize, although the post-$a$ process no longer has independent
  increments, the increment $X_{a+s}-X_a$ is still independent of $m$ given $a$.
  
  By steps 1 and 3, the increment $a ( x + h ) - a( x )$ is independent of $( a ( x - h ) - a ( x ) )_{h \geq 0}$
  given $a$ and $m$.  Since $a ( x + h ) - a ( x )$ only depends on $X$ through the increment $X_{a+s} - X_a$,
  given $a$ it is independent of $m$ as well. Therefore, $a ( x )$ is a Markov process and $u ( x )$ is a
  spectrally negative Markov process. Notice that if $\pot_0$ is not spectrally negative or does not have
  independent increments, then $a( x + h ) - a ( x )$ depends on $\pot_0(a_-)$ or $\pot_0 ( a )$ in
  addition to $a$, destroying the Markov property.
\end{proof}

%% file: bv.tex
\section{BV calculus and the Lax equations}
\label{sec:bv}
In this section, we derive the Lax equation \qref{eq:kinetic} using the conservation law~\qref{eq:sc1} and the \Volpert\/ chain rule for BV functions. This calculation is similar in spirit to \cite{EvE}.  We assume that for every $t>0$,  $u(x,t)$ is a spectrally negative, stationary Feller process.  The generator of $u$ is given by \qref{eq:gena1} and the operator $\genb$ is defined by \qref{eq:genb}. In addition, we define one and two-point operators as follows. We associate a linear functional $\PP$ to the stationary one-point distribution 
$p(y,t)$ via
\be
\label{onefn5}
\PP(\testfn) = \E\left(\testfn(u(x,t))\right) = \int_\R \testfn(y) p(y) \, dy. 
\ee
(Here and in what follows, it is convenient to suppress $t$ in the notation.)
Similarly, given $h >0$ we denote
the transition kernel for the process $u$ by $q_{h}$ and define the associated transition operator
\be
\label{onefn6}
 (\QQ_h \testfn) (y) = \int_\R q_h(y,z) \testfn(z) \, dz. 
\ee

\subsection{The \onefn\/}
We recall that the entropy solutions to \qref{eq:sc1} are in $\BVloc(\R_+\times \R)$.  \Volpert\/ showed that one may extend the chain rule to such functions in a natural manner. Let $u_\pm =u(x_\pm,t)$ denote the right and left limits of $u(x,t)$ and for any test function $\testfn$, consider the composition $\testfn(u(x,t))$ and set
\be
\label{eq:bv1}
[\testfn]=\int_0^1 \testfn\left(u_-+ \beta(u_+-u_-)\right)d\beta.
\ee
Then for every smooth test function with compact support, the entropy solution to \qref{eq:sc1} satisfies ~\cite[p.248]{Volpert}
\be
\label{eq:bv2}
\partial_t \testfn(u(x,t)) = -[\testfn][f] \partial_x u(x,t), \quad x \in \R, t>0.
\ee
The law of the \onefn\/ is determined by $\E\left(\testfn(u(x,t)\right)$ for arbitrary $\testfn$. We use equation \qref{eq:bv2} to obtain
\be
\label{onef1}
\partial_t \E \left( \testfn\left(u(x,t) \right) \right) = - \E \left(
  [\testfn][f] \partial_x u(x,t) \right).  
\ee
The left hand side of \qref{twof1} is $(\partial_t P) \testfn$. We must determine the right hand side. To this end, fix $x\in \R$ and for $h>0$ let  us denote $\umh=u(x-h,t)$. Our calculation relies on the following {\em unjustified\/} interchange of limits:
\be
\label{onef2}
\E \left([\testfn][f] \partial_x u(x,t) \right) = \lim_{h\dnto 0} \E
\left( \frac{1}{h} 
  \left(\testfn(\up)-\testfn(\umh)\right) \frac{\left(f(\up)
      -f(\umh)\right)}{\up-\umh} \right).
\ee
We assume \qref{onef2}, and show that
\be
\label{lim1}
\lim_{h\to 0} \E \left( \frac{1}{h}
  \left(\testfn(\up)-\testfn(\umh)\right) \frac{\left(f(\up)
      -f(\umh)\right)}{\up-\umh} \right) = -\PP \left(\genb \testfn\right),
\ee
As a consequence, 
\be
\label{onef3}
(\partial_t \PP) \testfn = \PP \genb \testfn,
\ee
which is the formal forward equation \qref{eq:forkolm}.

We now establish \qref{lim1} for a monomial $f(u)=u^m$. Since $\genb$ depends linearly on $f$, \qref{lim1} then also holds for polynomials and by approximation for $C^1$ fluxes with polynomial growth.  For $f(u)=u^m$ we expand the jump term  
\be
\label{fex}
\frac{f(\up)-f(\umh)}{\up -\umh} = \sum_{j=0}^{m-1} \umh^{m-1-j} \up^{j}. 
\ee
Therefore, the limit in \qref{lim1} is a sum of terms of the form
\ba
\nn
\lefteqn{ \left( \testfn(\up)-\testfn(\umh)\right) \umh^{m-1-j} \up^{j} = }
\\
\nn
&&
\umh^{m-1-j} \left(  \testfn(\up) \up^j - \testfn(\umh) \umh^j  -
\testfn(\umh) \left(\up^j -\umh^j \right) \right). 
\ea
The limit of each of these terms is
\ba
\nn
\lefteqn{ \lim_{h\to 0}\frac{1}{h} \E\left(  \left(
      \testfn(\up)-\testfn(\umh)\right) \umh^{m-1-j} 
  \up^j\right) = } 
\\
\nn
&&
\lim_{h \to 0} \frac{1}{h} \PP\left( y^{m-1-j}\left( \QQ_h(y^j
      \testfn) - y^j\testfn  - \testfn \left( \QQ_h(y^j)
      -y^j \right)\right) \right)\\
\label{gen1}
&&
= \PP \left( y^{m-1-j} \left( \gena(y^j \testfn) - \testfn \gena \left(y^j
    \right) \right) \right). 
\ea 
The term  $\gena(y^j \testfn) - \testfn \gena \left(y^j \right)$ may be computed
using \qref{eq:gena1}. First the drift term is
\be
\label{drift}
 y^{m-1-j} \left( \drift (y^j \testfn)' - \testfn \drift (y^j)'\right)= \drift y^{m-1}
\testfn'. 
\ee
Similarly, the jump term simplifies to
\be
\label{jump1}
y^{m-1-j} \int_\R n(y,z)  z^j \left( \testfn(z) - \testfn(y) \right) \, dz. 
\ee
We now sum over all terms  in the expansion \qref{fex} to obtain 
\ba
\nn
\lefteqn{\lim_{h\to 0} \E \left( \frac{1}{h}
  \left(\testfn(\up)-\testfn(\umh)\right) \frac{\left(f(\up)
      -f(\umh)\right)}{\up-\umh} \right)}\\
\nn
&& =  \sum_{j=0}^{m-1} \PP \left(y^{m-1-j} \left( \gena(y^j\testfn) -  \testfn \gena
  y^j \right) \right)\\
\nn
&&= \PP\left( \sum_{j=0}^{m-1} \drift y^{m-1}\testfn' + \int_\R n(y,z) y^{m-1-j}
  z^j \left( \testfn(z) - \testfn(y) \right) \, dz \right) \\ 
&& 
\nn
= \PP\left(  \drift my^{m-1}\testfn' + \int_\R n(y,z) \frac{z^m-y^m}{z-y} \left(
  \testfn(z) - \testfn(y) \right) \, dz\right) = -\PP\genb \testfn. 
\ea

\subsection{The \twofn\/}
Fix $x \in \R$ and $\alpha>0$. Let $\testfn$ and $\psi$ be two test functions. The law of the \twofn\/ is
described completely by 
\[ \E\left(\testfn\left(u(x,t)\right)
  \psi\left(u(x+\alpha,t)\right)\right) = \PP\left(\testfn \QQ_\alpha \psi
  \right). \]
The BV chain rule extends to a product rule as follows
\[ \partial_t \left( \testfn \QQ_\alpha\psi \right) =
\overline{\QQ_\alpha\psi} \partial_t \testfn + 
\overline{\testfn}  \partial_t \QQ_\alpha \psi, \]
where 
\[ \overline{\testfn}(x,t) =\frac{1}{2}\left( \testfn(x,t_-)+ \testfn(x,t_+)
\right). \]
We combine the BV chain rule and the conservation law to
obtain 
\[ \partial_t \testfn(u(x,t)) = -[\testfn][f] \partial_x
u(x,t),\quad  \partial_t\psi(u(x+\alpha,t)) = -[\psi][f] \partial_x
u(x+\alpha,t). \] 
Therefore,
\be
\label{twof1}
\partial_t \left(\testfn(u(x,t) \psi(x+\alpha,t) \right) =-\overline{\psi}
[\testfn'][f'] \partial_x u(x,t) - \overline{\testfn} [\psi'][f']\partial_x
u(x+\alpha,t). 
\ee
We compute the expected value of each of these terms in turn. First,
as earlier, the main assumption is
\ba
\label{twof2}
\lefteqn{ \E \left( \overline{\psi} [\testfn'][f'] \partial_x u(x,t) \right)
}\\
\nn
&&
= \lim_{h \to 0} \frac{1}{h} \E \left(
  \frac{f(\up)-f(\umh)}{\up-\umh}\left(\testfn(\up) -\testfn(\umh)\right)
  \left( \frac{\psip + \psimh}{2} \right) \right), 
\ea
where $\psip=\psi(u(x+\alpha,t))$ and $\psimh=\psi(u(x+\alpha-h,t))$. 

We can do away with the mean value. To be explicit, we write the above expectation as
\be
\label{twof3}
\int_\R\int_\R p(y) q_h(y,z) \left(
    \frac{f(z)-f(y)}{z-y} \left(\testfn(z)-\testfn(y)\right) \right)
\frac{\left(\QQ_{\alpha+h}+\QQ_\alpha\right) \psi(z)}{2}\,dz \, dy.
\ee
We see that the last term converges to $\QQ_\alpha \psi(y)$ as
$h \to 0$. Therefore, 
\ba
\nn
\lefteqn{ \E \left( \overline{\psi} [\testfn'][f'] \partial_x u(x,t) \right)}
\\
\nn &&
= \lim_{h \to 0} \frac{1}{h} \E \left(
  \frac{f(\up)-f(\umh)}{\up-\umh}\left(\testfn(\up) -\testfn(\umh)\right)
  \QQ_\alpha \psi(\up)\right) \\
\label{lim3} && 
= \lim_{h \to 0} \frac{1}{h} \E \left( \frac{f(\up)-f(\umh)}{\up-\um}
  \left( \testfn \QQ_\alpha \psi (\up) - \testfn \QQ_\alpha \psi(\umh)\right)
\right. \\
\label{lim4} &&
- \left. \testfn(\umh) \frac{f(\up)-f(\umh)}{\up-\umh}
  \left( \QQ_\alpha \psi (\up) - \QQ_\alpha \psi(\umh)\right)
\right), \\
\label{twof4}
&& = -\PP \left(\genb \left(\testfn \QQ_\alpha \psi\right) - \testfn \genb \QQ_\alpha
  \psi \right), 
\ea
where we used \qref{lim1} to compute the limits in \qref{lim3} and \qref{lim4}. 

We now compute the second term on the right hand side of
\qref{twof1}. As in the calculation above, we can do away with the
mean value, and we have
\ba
\nn
\lefteqn{ \E \left(\overline{\testfn} [\psi'][f'] \partial_x u(x+\alpha,t) \right)}
\\
\nn &&
= \lim_{h \to 0} \int_\R \int_\R \int_\R p(y) \testfn(y) q_{\alpha-h}
(y,z) 
\frac{q_h(z,w)}{h} \left( \frac{f(w)-f(z)}{w-z}
  \left(\psi(w)-\psi(z) \right) \right) \, dw \, dz\, dy. 
\ea
As in the computation of \qref{lim1} we find that the limit of the
innermost integral is $ \genb \psi(z)$. Therefore, 
\be
\label{twof6}
\E \left(\overline{\testfn} [\psi'][f'] \partial_x u(x+\alpha,t) \right) = -\PP\left(
  \testfn \QQ_\alpha \genb \psi \right). 
\ee
We combine \qref{twof1}, \qref{twof4} and \qref{twof6} to find
\be
\label{twof7}
\partial_t \left( \PP \left( \testfn \QQ_\alpha \psi \right) \right)=  \PP
\left(\genb \left(\testfn \QQ_\alpha \psi\right) -
\testfn \genb \QQ_\alpha   \psi  + \testfn \QQ_\alpha \genb \psi \right). 
\ee
But the left hand side is simply
\be
\label{twof8}
 (\partial_t \PP) \testfn \QQ_\alpha \psi + \PP\left( \testfn \partial_t
  \QQ_\alpha \psi \right), 
\ee
and by \qref{onef3}
\be
\label{twof9}
 (\partial_t \PP) \testfn \QQ_\alpha \psi = \PP \genb \left( \testfn \QQ_\alpha \psi \right). 
\ee 
We now combine \qref{twof7}, \qref{twof8} and \qref{twof9} to obtain
\be
\label{twof10}
\PP\left( \testfn \partial_t \QQ_\alpha \psi \right) = \PP \left(\testfn \left( \QQ_\alpha  \genb -\genb \QQ_\alpha \right) \psi \right). 
\ee
This equation holds for all $\testfn$ and $\psi$ in the domain of $\gena$. Thus, we
may write
\be
\label{twof11}
\partial_t \QQ_\alpha = \left( \QQ_\alpha  \genb - \genb  \QQ_\alpha \right) =
[\QQ_\alpha,\genb],
\ee
and in the limit $\alpha \dnto 0$ we have
\be
\label{twof12}
\partial_t \gena = [\gena,\genb]. 
\ee

%% file: hopf.tex
\section{Hopf's method}
In this section we derive the Lax equation \qref{eq:kinetic} following Hopf's method~\cite{Hopf-stat}. This method was used by Chabanol and Duchon to derive their kinetic equations. We show that the calculation works for any convex flux $f$. 

Though our calculations are mainly formal, we begin with the canonical framework of 
\S~\ref{sec:closure}.  Recall that the initial data $u_0$ has law $\mu_0$ on $\Omega$. Let $u ( t ; u_0 ) : [ 0, \infty ) \rightarrow \Omega$ be a weak solution of
(\ref{eq:sc1}) with initial data $u_0$, which induces a sequence
of probability measures $( \mu_t )_{t > 0}$ on $\Omega$. To derive an equation
for the flow of measures $( \mu_t )_{t > 0}$, make the assumption that $u$ is
differentiable in $t$ so that the conservation law can be written for all
$\varphi \in C_c^{\infty} ( \mathbb{R} )$ as
\begin{equation}
  \label{eq:HopfDynamicalSystem} \partial_t \langle u ( t ; u_0 ), \varphi
  \rangle - \left\langle f ( u ( t ; u_0 ) ), \varphi' \rangle = 0. \right.
\end{equation}
Here, $' = \partial_x$ and $\langle \cdot, \cdot \rangle$ is the standard
duality pairing.  Define the Hopf characteristic functional $\hat{\mu}_t$ of the law $\mu_t$
\begin{equation}
  \hat{\mu}_t ( \varphi ) = \int_{\Omega} e^{i \langle u, \varphi \rangle}
  \mu_t ( d u )
\end{equation}
which evolves according to
\begin{align}
  \partial_t \hat{\mu}_t ( \varphi ) = \partial_t \int_{\Omega} e^{i \langle
  u, \varphi \rangle} \mu_t ( d u ) & = \partial_t \int_{\Omega} e^{i
  \langle u ( t ; u_0 ), \varphi \rangle} \mu_0 ( d u_0 ) \nonumber\\
  & = \int_{\Omega} i \left\langle f ( u ( t ; u_0 ), \varphi'
  \right\rangle e^{i \langle u ( t ; u_0 ), \varphi \rangle} \mu_0 ( d u_0 )
  \nonumber\\
  & = \int_{\Omega} i \left\langle f ( u ), \varphi' \right\rangle e^{i
  \langle u, \varphi \rangle} \mu_t ( d u ) . 
  \label{eq:HopfFunctionalEvolution}
\end{align}
Any set of probability measures $( \mu_t )_{t > 0}$ on $\Omega$ that satisfies
(\ref{eq:HopfFunctionalEvolution}) for all $\varphi \in C_c^{\infty} (
\mathbb{R} )$ is defined to be a {\emph{statistical solution}} of the scalar
conservation law---in particular, the flow of measures generated by the
entropy solution is a statistical solution. Assuming that all moments of
$\mu_t$ are finite, the exponential can be expanded in series (denoting
$d\mathbf{x}_n = \Pi_{j = 1}^n d x_j$) as
\[ e^{i \langle u, \varphi \rangle} = \sum_{n = 0}^{\infty} \frac{i^n}{n!}
   \int_{\mathbb{R}^n} \prod_{j = 1}^n u ( x_j ) \varphi ( x_j )
   d\mathbf{x}_n . \]
Let
\[ \mathbb{E}_{\mu_t} \left[ \prod_{i = 1}^n \testfn_i ( u ( x_i ) ) \right] =
   \int_{\Omega} \prod_{i = 1}^n \testfn_i ( u ( x_i ) ) \mu_t ( d u ) . \]
Substituting the expansion for the exponential into
(\ref{eq:HopfFunctionalEvolution}) yields the infinite hierarchy
\begin{align}
  & \sum_{n = 0}^{\infty} \frac{i^n}{n!} \int_{\mathbb{R}^n} \partial_t
  \mathbb{E}_{\mu_t} \left[ \prod_{j = 1}^n u ( x_j ) \varphi ( x_j ) \right]
  d\mathbf{x}_n \nonumber\\
  & = \sum_{n = 0}^{\infty} \frac{i^{n + 1}}{n!} \int_{\mathbb{R}^{n +
  1}} \mathbb{E}_{\mu_t} \left[ f ( u ) \varphi' ( x ) \prod_{j = 1}^n u ( x_j )
  \varphi ( x_j ) \right] d x d\mathbf{x}_n .  \label{eq:HopfHierarchy1}
\end{align}
We simplify this hierarchy as follows. If $g:\R^n\to \R$ satisfies 
   $g ( x_1, \ldots, x_n ) = g ( x_{\sigma ( 1 )}, \ldots x_{\sigma ( n )}
  )$ for all permutations $\sigma$ of $\{ 1, \ldots, n \}$, then
  \begin{equation}
    \label{eq:CalculusTrick} \int_{\mathbb{R}^n} g ( x_1, \ldots, x_n )
    d\mathbf{x}_n = n! \int_{x_1 < x_2 < \cdots < x_n} g ( x_1, \ldots, x_n
    ) d\mathbf{x}_n .
  \end{equation}
We choose
\[ g ( x_1, \ldots, x_{n + 1} ) = \frac{1}{n} \sum_{k = 1}^{n + 1} \mathbb{E}_{\mu_t}
   \left[ f ( u ( x_k ) ) \varphi' ( x_k ) \prod_{j \neq k} u ( x_j ) \varphi
   ( x_j ) \right], \]
in (\ref{eq:CalculusTrick}) and substitute in (\ref{eq:HopfHierarchy1}) to
obtain
\begin{align}
  & \sum_{n = 0}^{\infty} i^n \int_{x_1 < \cdots < x_n} \partial_t
  \mathbb{E}_{\mu_t} \left[ \prod_{j = 1}^n u ( x_j ) \varphi ( x_j ) \right]
  d\mathbf{x}_n \nonumber\\
  & = \sum_{n = 0}^{\infty} i^{n + 1} \int_{x_1 < \cdots < x_{n + 1}}
  \sum_{k = 1}^{n + 1} \mathbb{E}_{\mu_t} \left[ f ( u ( x_k ) ) \varphi' ( x_k )
  \prod_{j \neq k} u ( x_j ) \varphi ( x_j ) \right] d\mathbf{x}_n . 
  \label{eq:HopfHierarchy2}
\end{align}
Re-indexing the l.h.s. of (\ref{eq:HopfHierarchy2}) and subtracting the r.h.s.
gives 
\begin{align}
  \sum_{n = 1}^{\infty} i^n \int_{x_1 < \cdots < x_n} \left\{ \right. &
  \partial_t \mathbb{E}_{\mu_t} \left[ \prod_{j = 1}^n u ( x_j ) \right] \prod_{j =
  1}^n \varphi ( x_j ) \nonumber\\
  & \left. - \sum_{k = 1}^n \mathbb{E}_{\mu_t} \left[ f ( u ( x_k ) ) \prod_{j \neq
  k} u ( x_j ) \right] \varphi' ( x_k ) \prod_{j \neq k} \varphi ( x_j )
  \right\} d\mathbf{x}_n = 0.  \label{eq:HopfHierarchy3}
\end{align}

Finally, the statistical hierarchy (\ref{eq:HopfHierarchy3}) can be
considerably simplified with the closure assumption that $( \mu_t )_{t
> 0}$ is the law of a stationary Feller process with the one and two-point operators
$\mathcal{P}$ and $( \mathcal{Q}_{h} )_{h > 0}$ defined in \qref{onefn5} and \qref{onefn6}. With $h_i = x_i - x_{i - 1}$ and $\mathcal{Q}_i = \mathcal{Q}_{h_i}$ for $i = 2, \ldots, n$,
\begin{align}
  \mathbb{E}_{\mu_t} & \left[ \prod_{i = 1}^n \testfn_i ( u ( x_i ) ) \right]
  \nonumber\\
  & = \int_{\mathbb{R}} p ( d u_1, t ) \testfn_1 ( u_1 ) \int_{\mathbb{R}}
  q_{h_2} ( u_1, d u_2, t ) \testfn_2 ( u_2 ) \cdots \int_{\mathbb{R}} q_{h_n} (
  u_{n - 1}, d u_n, t ) \testfn_n ( u_n ) \nonumber\\
  & = \mathcal{P} \testfn_1 \mathcal{Q}_2 \testfn_2 \cdots \mathcal{Q}_n \testfn_n
  . \nonumber
\end{align}
Assume the transition measures $q_h$ are differentiable in $h$. Transferring
the derivative on the test function $\varphi' ( x_k )$ onto $q_{h_k} ( u_{k -
1}, d u_k )$ in (\ref{eq:HopfHierarchy2}) and using integration by parts, the
statistical hierarchy is
\begin{align}
  \sum_{n = 1}^{\infty} i^n \int_{x_1 < \cdots < x_n} \left\{ \right. &
  \partial_t \mathbb{E}_{\mu_t} \left[ \prod_{j = 1}^n u ( x_j ) \right] \nonumber\\
  & \left. + \sum_{k = 1}^n \partial_{x_k} \mathbb{E}_{\mu_t} \left[ f ( u ( x_k )
  ) \prod_{j \neq k} u ( x_j ) \right] \right\} \prod_{j = 1}^n \varphi ( x_j
  ) d\mathbf{x}_n = 0. \nonumber
\end{align}
\[  \]
Note that the boundary terms in the previous equation vanish due to
cancellation of terms and the fact that $\varphi \in C_c^{\infty} (
\mathbb{R} )$ has compact support. The density of tensor products of the form
$\phi ( x_1, \ldots, x_n ) = \prod_{j = 1}^n \varphi ( x_j )$ in the space of
test functions on $\{ x = ( x_1, \ldots, x_n ) \in \mathbb{R}^n : x_1 <
\cdots < x_n \}$ implies that for $n \in \mathbb{N}$ and $x_1 < x_2 < \cdots
< x_n$, the following infinite set of equations hold:
\begin{equation}
  \partial_t \mathbb{E}_{\mu_t} \left[ \prod_{j = 1}^n u ( x_j ) \right] = - \sum_{k =
  1}^n \partial_{x_k} \mathbb{E}_{\mu_t} \left[ f ( u ( x_k ) ) \prod_{j \neq k} u (
  x_j ) \right] .
\end{equation}
In terms of the operators $\mathcal{Q}_i' = \partial_{h_i} \mathcal{Q}_i$
(denoting $h_1 = x_1$, $\mathcal{Q}_1 = \mathcal{P}$, $\mathcal{Q}_{n + 1} =
I$) this is
\begin{align}
  \partial_t \mathbb{E}_{\mu_t} \left[ \prod_{j = 1}^n u ( x_j ) \right] = \sum^n_{k =
  1} \{ & \mathcal{P} u \mathcal{Q}_2 u \cdots \mathcal{Q}_k f ( u )
  \mathcal{Q}'_{k + 1} u \cdots \mathcal{Q}_n u \nonumber\\
  & - \mathcal{P} u \mathcal{Q}_2 u \cdots \mathcal{Q}'_k f ( u )
  \mathcal{Q}_{k + 1} u \cdots \mathcal{Q}_n u \}, \nonumber
\end{align}
or equivalently, in terms of the generator $\mathcal{A}$ (using
$\mathcal{Q}_i' = \mathcal{A} \mathcal{Q}_i$)
\[ \partial_t \mathbb{E}_{\mu_t} \left[ \prod_{j = 1}^n u ( x_j ) \right] = \sum^n_{k =
   1} \mathcal{P} u \mathcal{Q}_2 u \cdots \mathcal{Q}_k \left[ f ( u ),
   \mathcal{A} \right] \mathcal{Q}_{k + 1} u \cdots \mathcal{Q}_n u. \]
One obtains the evolution equation for the 1-point function by taking $h_i
\rightarrow 0$ for all $i$ (i.e., $\mathcal{Q}_i = I$ for all $i$), or the
equation for the 2-point function by taking $h_{l + 1} = \alpha$ for $l \in \{
1, \ldots, n - 1 \}$ and $h_i \rightarrow 0$ for all $i \neq l + 1$ (i.e.,
$\mathcal{Q}_{l + 1} = \mathcal{Q}_{\alpha}$ and $\mathcal{Q}_i = I$ for all
$i \neq l + 1$):
\begin{equation}
  \label{eq:Hopf1-PointSum} \partial_t \mathbb{E}_{\mu_t} \left[ ( u^{} ( x_1 ) )^n
  \right] = \sum^n_{k = 1} \mathcal{P} u^{k - 1} \left[ f ( u ), \mathcal{A}
  \right] u^{n - k}
\end{equation}
\begin{align}
  \partial_t \mathbb{E}_{\mu_t} \left[ ( u ( x_1 ) )^l ( u ( x_1 + \alpha ) )^{n - l}
  \right] = & \sum^l_{k = 1} \mathcal{P} y^{k - 1} \left[ f ( y ),
  \mathcal{A} \right] y^{l - k} \mathcal{Q}_{\alpha} y^{n - l} \nonumber\\
  & + \sum^n_{k = l + 1} \mathcal{P} y^l \mathcal{Q}_{\alpha} y^{k - ( l +
  1 )} \left[ f ( y ), \mathcal{A} \right] y^{n - k} 
  \label{eq:Hopf2-PointSum}
\end{align}
Using the expression \qref{eq:gena1} for the generator $\mathcal{A}$,
(\ref{eq:Hopf1-PointSum}) and (\ref{eq:Hopf2-PointSum}) simplify:
\begin{align}
  \partial_t \mathcal{P} y^n & = - \sum_{k = 1}^n \mathcal{P} y^{k - 1}
  \left\{ \drift ( y ) f' ( y ) y^{n - k} + \int_{\mathbb{R}} \jump ( y, dz ) ( f ( z
  ) - f ( y ) ) z^{n - k} \right\} \nonumber\\
  & = - \mathcal{P} \left\{ \drift ( y ) f' ( y ) ( y^n )' + \int_{\mathbb{R}}
  \jump ( y, d z ) \frac{f ( z ) - f ( y )}{z - y} ( z^n - y^n ) \right\}
  \nonumber
\end{align}
\begin{align}
  & \partial_t ( \mathcal{P} y^l \mathcal{Q} y^{n - l} ) \nonumber\\
  & = - \sum_{k = 1}^l \mathcal{P} y^{k - 1} \left\{ \drift ( y ) f' ( y ) y^{l
  - k} \mathcal{Q}_{\alpha} y^{n - l} + \int_{\mathbb{R}} \jump ( y, d z ) ( f (
  z ) - f ( y ) ) z^{l - k} \mathcal{Q}_{\alpha} z^{n - l} \right\}
  \nonumber\\
  & - \sum_{k = l + 1}^n \mathcal{P} y^l \mathcal{Q}_{\alpha} y^{k - ( l +
  1 )} \left\{ \drift ( y ) f' ( y ) y^{n - k} + \int_{\mathbb{R}} \jump ( y, d z ) (
  f ( z ) - f ( y ) ) z^{n - k} \right\} \nonumber\\
  & = - \mathcal{P} \left\{ \drift ( y ) f' ( y ) ( y^l )' \mathcal{Q}_{\alpha}
  y^{n - l} + \int_{\mathbb{R}} \jump ( y, d z ) \frac{f ( z ) - f ( y )}{z - y}
  ( z^l - y^l ) \mathcal{Q}_{\alpha} z^{n - l} \right\} \nonumber\\
  & - \mathcal{P} y^l \mathcal{Q}_{\alpha} \left\{ \drift ( y ) f' ( y ) ( y^{n
  - l} )' + \int_{\mathbb{R}} \jump ( y, d z ) \frac{f ( z ) - f ( y )}{z - y} (
  z^{n - l} - y^{n - l} ) . \right\} \nonumber
\end{align}
In terms of the operator $\mathcal{B}$ defined in \qref{eq:genb}
the above expressions read
\[ \partial_t \mathcal{P} y^n = \mathcal{P} \mathcal{B} y^n, \quad
   \partial_t ( \mathcal{P} y^l \mathcal{Q}_{\alpha} y^{n - l} ) = 
   \mathcal{P} ( \mathcal{B} ( y^l \mathcal{Q}_{\alpha} y^{n - l} ) - y^l
   \mathcal{B} \mathcal{Q}_{\alpha} y^{n - l} + y^l \mathcal{Q}_{\alpha}
   \mathcal{B} y^{n - l} ) . \]
Since $\mathcal{B}$ is a linear operator, we approximate $\testfn, \psi \in C_c^{\infty} ( \mathbb{R} )$ by polynomials to obtain the evolution
equation for the 1- and 2-point operators: 
\[ \partial_t \mathcal{P} \testfn = \mathcal{P} \mathcal{B} \testfn, \qquad
   \partial_t ( \mathcal{P} \testfn \mathcal{Q}_{\alpha} \psi ) = \mathcal{P} (
   \mathcal{B} ( \testfn \mathcal{Q}_{\alpha} \psi ) - \testfn \mathcal{B}
   \mathcal{Q}_{\alpha} \psi + \testfn \mathcal{Q}_{\alpha} \mathcal{B} \psi ) .
\]
Therefore we have the evolution $\partial_t \mathcal{A} = [ \mathcal{A},
\mathcal{B} ]$ exactly as in \S~\ref{sec:bv}.

%% file: sss.tex
\section{Groeneboom's solution}
\label{sec:sss}
In this section we verify that the generator
\be
\label{eq:sss1}
\gena(t) \testfn (y) = \frac{1}{t}\testfn'(y) + \int_{-\infty}^y \frac{1}{t^{1/3}}\jumpex(yt^{1/3},zt^{1/3}) \left(\testfn(z)-\testfn(y)\right)\, dz
\ee
satisfies the Lax equation \qref{eq:kinetic} when $f(u)=u^2/2$ and $\jumpex$ is given by \qref{eq:groen2}--\qref{eq:groen4}. In this case, the evolution of the drift \qref{eq:evolb} is simply $\dot\drift =-\drift^2$. It is clear that $\drift(y,t)= t^{-1}$ is a solution. The equation for the jump density now takes the form
\be
\label{eq:sss2}
\partial_t \jump(y,z,t)  -\frac{1}{2t}(y-z)\left(\partial_y \jump -\partial_z \jump\right) = Q(\jump,\jump) 
\ee
with the collision kernel
\ba
\label{eq:sss3}
Q(\jump,\jump)(y,z,t) &&= \frac{y-z}{2}\int_z^y \jump(y,w,t)\jump(w,z,t)\, dw\\
\nonumber
&& 
- \jump(y,z,t)\int_{-\infty}^z \frac{y-w}{2}\jump(z,w,t)\, dw\\
\nonumber
&& -\jump(y,z,t)\int_{-\infty}^y \frac{w-z}{2} \jump(y,w,t)\jump(y,w,t)\, dw.
\ea
We substitute the ansatz 
\be
\label{eq:sss4}
u=yt^{1/3}, \quad v =zt^{1/3},\quad s=u-v, \quad \jump(y,z,t) = t^{-1/3}\jumpex(u,v),
\ee
in \qref{eq:sss2} and collect terms. We must then verify that
\ba
\label{eq:sss5}
\lefteqn{-\frac{2}{3} + \left(u-\frac{v}{3}\right)\frac{J'(v)}{J(v)} + \left(\frac{u}{3}-v\right)\frac{J'(u)}{J(u)} -\frac{4}{3}\frac{sK'(s)}{K(s)}}\\
\nonumber
&&
=s\frac{K*K}{K}(s) -\left(s \frac{K*J}{J}(v) + \frac{J*(xK)}{J}(v)) \right) + \left(\frac{J*(xK)}{J}(u) - s \frac{K*J}{J}(u)\right).
\ea
Here $x$ plays the role of a dummy variable in  the convolutions $J*(xK)$. Explicitly, $J*(xK)(u) = \int_{\R} J(u-x) xK(x) \, dx$. Note also that the terms in  \qref{eq:sss5} are evaluated at the arguments $u$, $v$ or $s=u-v$ respectively. The functions $J$ and $K$ are related by the following identities:
\ba
\label{eq:m1}
x^2J  &=& K*J + J' \\
\label{eq:m2}
\frac{3}{2}J*(xK) &=& x\left(K*J\right) + J \\
\label{eq:m3} 
x^3K & = &  3x(K*K) + 4x K' + 2K.
\ea
(The argument of each function in these identities is $x$). We substitute these identities in \qref{eq:sss5}, and collect the three terms corresponding to the arguments $u$,$v$ and $s$ and find that (remarkably!) each of them reduces to a polynomial, and the sum of these three polynomials vanishes.

The identities \qref{eq:m1}--\qref{eq:m3} are proven using the Laplace transform, \qref{eq:groen3}, and the definition $\Ai''(q) =q \Ai(q)$. First, we observe that $l =j'/j =-\Ai'/\Ai$ solves the Riccati equation
\be
\label{eq:m5}
l' = -q + l^2.
\ee
As a consequence of the definition of $k$ in \qref{eq:groen3} we have
\be
\label{eq:m5a}
l'=k/2.
\ee
Therefore,  we also have
\be
\label{eq:m6}
k'=2l'' = -2 + 4 ll' = -2(1-lk).
\ee
Thus, $l$ and $k$ solve an autonomous system. Equation \qref{eq:m5} may be rewritten in terms of $j$ and $k$ as
\be
\label{eq:m4}
j'' -(q + k)j =0.
\ee
This is equivalent to \qref{eq:m1}. Next, if we rewrite \qref{eq:m6} in terms of $j$ and $k$, we find immediately that
\be
\label{eq:m7}
\frac{3}{2}jk' =(jk)'-j.
\ee
This is equivalent to \qref{eq:m2}. Finally, we differentiate \qref{eq:m6} twice and use \qref{eq:m5} to eliminate $l$ to find 
\ba
\label{eq:m8}
k'' &=& k^2 -4l(1-lk)\\
\label{eq:m9}
k''' &=& 3(k^2)' + 4qk'+ 2k,
\ea
which is equivalent to \qref{eq:m3}. 

Here is the \Painleve\/ property: differentiate \qref{eq:m5} to find $l''=2l^3 -2lq-1$. Now rescale $\tau=-q2^{1/3}$, $w(\tau)=2^{-1/3} l(q)$ to see that $w$ solves the second \Painleve\/ equation with parameter $1/2$~\cite{Ablowitz}.
\be
\label{eq:P2}
\frac{d^2w}{d\tau^2} = 2w^3 + w\tau + \frac{1}{2}.
\ee